\documentclass[preprint,12pt]{elsarticle}

\usepackage{algorithm}
\usepackage{algpseudocode}
\usepackage{amsfonts,amsmath,amssymb}
\usepackage{booktabs}
\usepackage{cancel}
\usepackage[skip=0pt]{caption}
\usepackage{graphicx}
\usepackage{mleftright}
\usepackage{subcaption}
\usepackage{multirow}
\usepackage{orcidlink}
\usepackage{placeins}
\usepackage{xcolor}

\algrenewcommand{\algorithmiccomment}[1]{\hfill$\triangleright$~#1}
\mleftright

\newdefinition{definition}{Definition}
\newtheorem{lemma}{Lemma}
\newtheorem{proposition}{Proposition}
\newtheorem{corollary}{Corollary}
\newproof{proof}{Proof}

\DeclareMathOperator*{\argmax}{arg\,max}

\newif\ifshowreview
\showreviewtrue
\definecolor{reviewcolor}{rgb}{0.00,0.25,0.70}

\begin{document}

\begin{frontmatter}

\title{Global and Local Entropy Tests for STL Files Under IEEE-754 Validity Constraints}

\author[shu]{Tonislav Troev\orcidlink{0009-0002-8803-5162}}
\ead{tonyslav\_troeff@outlook.com}

\author[shu]{Borislav Stoyanov\orcidlink{0000-0002-7307-5914}\corref{cor1}}
\ead{borislav.stoyanov@shu.bg}

\cortext[cor1]{Corresponding author.}

\affiliation[shu]{organization={Department of Computer Informatics, Konstantin Preslavsky University of Shumen},
            addressline={115 Universitetska Street},
            postcode={9700},
            city={Shumen},
            country={Bulgaria}}

\begin{abstract}
Shannon entropy is a standard measure of ciphertext randomness, evaluated either globally over the byte histogram of the entire ciphertext or locally as the average over non-overlapping blocks of fixed size.
While the global statistic captures only the overall byte distribution, the local statistic is also sensitive to non-randomness that is confined to parts of the ciphertext.
Both measures are commonly assessed against reference values derived under the assumption of independent uniform bytes.
However, STL encodes most numerical data as single-precision IEEE-754 floats, and format validity requires them to be finite, causing nonuniform byte probabilities and dependence between the two bytes of each exponent pair.
We derive the exact mean and variance of entropy under the valid-STL (VSTL) distribution, which assigns equal probability to every admissible sequence of triangle records, and provide second-order expansions that keep the global reference values practical to compute at large file sizes.
Using these expansions, we establish a two-sided hypothesis test for the global entropy statistic, whose limiting value is lowered from $8$ to $7.99991479609$ bits.
The format constraints also change the local entropy reference: applying the capacity energy criterion to triangle records selects a block size of $1950$ bytes instead of $1936$ bytes, with an exact reference mean of $7.903106732975$ rather than $7.902469316647$ bits.
Although small, these shifts are sufficient to place the VSTL reference means outside the conventional $95\%$ acceptance intervals at practical file sizes, making unadjusted entropy tests prone to falsely rejecting ideal VSTL ciphertexts.
Monte-Carlo simulations confirm all reference values within sampling error.
We therefore recommend using the format-specific reference values derived here to evaluate STL ciphers.
\end{abstract}


\end{frontmatter}

\section{Introduction}\label{sec:intro}

Additive manufacturing uses digital three dimensional (3D) models to produce physical parts, from prototypes to aerospace components and medical devices~\cite{gibson_2021}.
As their geometry may disclose proprietary information, protecting design data is part of securing the manufacturing process, in which theft of technical data is a principal threat~\cite{yampolskiy_2018}.
STL represents the geometry as a triangulated surface and remains a common exchange format in this workflow~\cite{gibson_2021,szilvasi_nagy_2003}.

The statistical evaluation of encryption methods for 3D models commonly uses measures adapted from image encryption.
Our foundational work~\cite{adjusted_npcr_uaci_2026} examined this transfer for NPCR and UACI and showed that their conventional reference values require adjustment for STL validity constraints.
The same issue arises for Shannon entropy~\cite{shannon_entropy_1948}: its reference values must reflect the admissible ciphertext encodings.

Global entropy is computed from the byte histogram of the entire ciphertext, whereas local entropy~\cite{local_entropy_2013} averages the same measure over disjoint blocks of fixed size and is also sensitive to departures from randomness confined to parts of the ciphertext.
The conventional reference values for both measures assume independent and uniformly distributed bytes.

In practice, however, evaluations of the global measure often report values close to the theoretical maximum of $8$ bits, without specifying acceptance bounds for the given sample size.
Such proximity alone does not establish consistency with the assumed null model.
A two-sided hypothesis test provides a more rigorous assessment because it rejects values that are either unusually low or unusually high.

Even under the uniform model, the expected entropy is strictly below the theoretical maximum for every finite sample size --- classical estimation results quantify the bias and variability~\cite{miller_1955,basharin_1959,harris_1977}.

As established in~\cite{adjusted_npcr_uaci_2026}, a valid STL file does not satisfy the assumption of independent uniform bytes because each triangle record is composed mostly of IEEE-754 single-precision floats~\cite{ieee754_2019}, for which the validity constraint excludes NaN and $\pm\infty$.
Consequently, the valid-STL (VSTL) distribution is uniform over admissible encodings and has nonuniform byte probabilities that vary by position, together with dependence between the exponent bytes of each float.

This paper therefore introduces a hypothesis test for the global entropy of valid STL files, adjusts the established local test, and provides reference values and acceptance bounds appropriate to the assumed distribution and sample size for both measures.

We derive exact global and local entropy moments from the individual and joint laws of byte counts in triangle records under the VSTL model.
For global entropy, second-order expansions make the reference moments efficiently computable for large files.
For local entropy, we determine the optimal block size under the VSTL model by adapting the capacity energy criterion~\cite{local_entropy_2013}.
Normal approximations provide two-sided acceptance bounds for both tests.

Although the validity constraint induces seemingly small shifts in the entropy references, these differences become practically relevant for both tests as their acceptance intervals narrow.
For local entropy, the change in block size further accentuates this effect.

The remainder of the paper is organized as follows.
Section~\ref{sec:prelim} defines the notation used throughout the paper and recalls the VSTL distribution from our foundational work.
Section~\ref{sec:block} derives the exact block entropy moments.
Section~\ref{sec:global} develops the global entropy test, while Section~\ref{sec:local} adapts the local test; both examine their practical significance and validate their corresponding reference values through Monte-Carlo simulations.
Finally, Section~\ref{sec:conclusions} concludes the paper.

\section{Preliminaries and Notation}\label{sec:prelim}

Most of the notation and structural facts in this section follow our foundational work~\cite{adjusted_npcr_uaci_2026}; we restate only those needed for the entropy analysis.

\subsection{The Byte Alphabet and Its Named Subsets}

Write $\mathcal{V} = \{0, \ldots, 255\}$ for the byte alphabet, so that $\mathcal{V}^k$ is the set of $k$-byte sequences.
Two pairs of subsets of $\mathcal{V}$ recur throughout, the \emph{halves} and the \emph{majority sets}:
\begin{equation}\label{eq:subsets}
\begin{aligned}
\mathcal{V}_{lo} &= \{0, \ldots, 127\}, &
\mathcal{V}^{*}_{lo} &= \mathcal{V}_{lo} \setminus \{127\}, \\
\mathcal{V}_{hi} &= \{128, \ldots, 255\}, &
\mathcal{V}^{*}_{hi} &= \mathcal{V}_{hi} \setminus \{255\}.
\end{aligned}
\end{equation}
The two majority sets drop $127$ and $255$, respectively; these values are the \emph{minority bytes}.

\subsection{Validity Constraint}\label{subsec:validity}

A valid float is any finite IEEE-754 single-precision 32-bit number~\cite{ieee754_2019} whose exponent field is not \texttt{0xFF}; subnormals and signed zeros are retained, NaN and $\pm\infty$ excluded.
The set $\mathcal{F} \subset \{0,1\}^{32}$ of valid encodings has cardinality $\lvert \mathcal{F} \rvert = 2^{32} - 2^{24} = 4278190080$.

Each float $(s,\, e_0 \ldots e_7,\, m_0 \ldots m_{22})$ splits into four bytes, from most to least significant:
\begin{equation}\label{eq:bytes}
B_0 = s \| e_{0 \ldots 6}, \quad B_1 = e_7 \| m_{0 \ldots 6}, \quad B_2 = m_{7 \ldots 14}, \quad B_3 = m_{15 \ldots 22}.
\end{equation}
Since $e \neq \texttt{0xFF}$ reads $\neg\,(e_{0 \ldots 6} = \texttt{0x7F} \wedge e_7 = 1)$, the validity constraint only restricts $B_0$ and $B_1$ jointly, while the other bytes remain unaffected.
Accordingly, bytes fall into three classes based on their position --- $U$, $B_0$ and $B_1$, as defined in~\cite[Sec.~2.3]{adjusted_npcr_uaci_2026}.

\subsection{The Valid-STL (VSTL) Distribution}\label{subsec:layout}

The STL file format allows representing triangulated surfaces in a readable ASCII encoding or a compact binary encoding.
We treat these representations as semantically equivalent and use the binary encoding for the subsequent analysis, converting ASCII content when necessary~\cite[Sec.~2.1]{adjusted_npcr_uaci_2026}.
Each triangle record contains a normal vector ($3$ floats, $12$ bytes), three vertices ($9$ floats, $36$ bytes), and an attribute word ($2$ bytes), for a total of $50$ bytes.
Throughout the rest of the paper a \emph{block} is a string of triangle records.

We adopt assumptions (A1) and (A2) of~\cite[Sec.~2.3]{adjusted_npcr_uaci_2026}.
Assumption (A1) concerns format-preserving encryption~\cite{bellare_fpe_2009}, while (A2) confines dependence to the exponent pair $(B_0, B_1)$ within each float.

We restrict our attention to the triangle records of the ciphertext; file metadata is outside the scope of the presented analysis.

\begin{definition}\label{def:vstl}
For a fixed number of triangles $T$, the \emph{valid-STL} (VSTL) distribution assigns equal probability to every valid STL file.
Floats and attribute words are sampled independently and uniformly over their respective admissible encodings, giving $26T$ unconstrained bytes and $12T$ exponent pairs that form mutually independent groups.
\end{definition}

\subsection{Marginal and Joint Byte Distributions}\label{subsec:marginals}

The marginals for the three byte classes~\cite[Sec.~3]{adjusted_npcr_uaci_2026} are
\begin{gather}
P_{B_U}(v) = \frac{1}{256}, \label{eq:pu} \\
P_{B_1}(v) =
\begin{cases}
P_{loB_1} = 1/255, & v \in \mathcal{V}_{lo}, \\
P_{hiB_1} = 127/32640, & v \in \mathcal{V}_{hi},
\end{cases} \label{eq:pb1} \\
P_{B_0}(v) =
\begin{cases}
P_{B_0^{*}} = 1/255, & v \in \mathcal{V}^{*}_{lo} \cup \mathcal{V}^{*}_{hi}, \\
1/510, & v \in \{127, 255\}.
\end{cases} \label{eq:pb0}
\end{gather}

The admissible exponent pairs form
\begin{equation}\label{eq:X}
\mathcal{X} = \left\{ (v_0, v_1) \in \mathcal{V}^2 : v_0 \in \mathcal{V}^{*}_{lo} \cup \mathcal{V}^{*}_{hi} \vee v_1 \in \mathcal{V}_{lo} \right\},
\end{equation}
on which $(B_0, B_1)$ is uniform:
\begin{equation}\label{eq:pb01}
P_{B_{0,1}}(v_0, v_1) =
\begin{cases}
\dfrac{1}{65280}, & (v_0, v_1) \in \mathcal{X}, \\[4pt]
0, & (v_0, v_1) \notin \mathcal{X}.
\end{cases}
\end{equation}
The joint law makes explicit the dependence imposed by the validity constraint of Section~\ref{subsec:validity} and permitted by (A2), since $P_{B_{0,1}} \neq P_{B_0} P_{B_1}$.

\subsection{Entropy Measures}

For a sequence $S \in \mathcal{V}^n$ in which the value $v$ occurs $\#(v)$ times, the Shannon entropy~\cite{shannon_entropy_1948} is
\begin{equation}\label{eq:plugin}
\begin{gathered}
H(S) = \sum_{v \in \mathcal{V}} \varphi \left(\frac{\#(v)}{N}\right), \\
\varphi(x) = x \lambda(x), \\
\lambda(x) =
\begin{cases}
0, & x = 0, \\
- \log_2 \left( x \right), & x > 0.
\end{cases} 
\end{gathered}
\end{equation}

We reserve $H_\infty$ for the limiting value of $H$ as the string length grows.
For a source of independent uniform bytes $H_\infty^{\mathrm{uniform}} = \log_2 \lvert \mathcal{V} \rvert = 8$ bits.
Section~\ref{subsec:ceiling} quantifies the reduction imposed by the validity constraint and establishes the corresponding entropy ceiling for the VSTL model.

This paper considers two entropy measures used to evaluate cryptographic systems.
The \emph{global} measure evaluates~\eqref{eq:plugin} on the entire ciphertext.
The \emph{local} measure~\cite{local_entropy_2013} averages~\eqref{eq:plugin} over $k$ non-overlapping blocks of $N_B$ bytes each.
Sections~\ref{sec:global} and~\ref{sec:local}, respectively, provide reference values, acceptance bounds and test procedures specific to the VSTL format.

\section{Exact Entropy Statistics of a Valid STL File}\label{sec:block}

This section derives the exact mean and variance of the block entropy under the VSTL distribution, accounting for the dependence within each exponent pair.

\subsection{The Block Entropy as a Function of Its Histogram}

The moments of the block entropy can be expressed in terms of the byte counts:
\begin{equation}\label{eq:muterms}
\mu^{\mathrm{VSTL}}(T) = \sum_{v \in \mathcal{V}} E \left[ \varphi \left( \frac{\#(v)}{N} \right) \right],
\end{equation}
\begin{equation}\label{eq:varterms}
\begin{aligned}
\left[ \sigma^{\mathrm{VSTL}}(T) \right]^2 &= \sum_{v \in \mathcal{V}} E \left[ \varphi \left( \frac{\#(v)}{N} \right)^2 \right] \\
&\quad + \sum_{\substack{v_0, v_1 \in \mathcal{V} \\ v_0 \neq v_1}} E \left[ \varphi \left( \frac{\#(v_0)}{N} \right) \varphi \left( \frac{\#(v_1)}{N} \right) \right] \\
&\quad - \left( \sum_{v \in \mathcal{V}} E \left[ \varphi \left( \frac{\#(v)}{N} \right) \right] \right)^2.
\end{aligned}
\end{equation}

Since $H$ is additive over byte values, only the one-dimensional and two-dimensional marginals of $\#$ ever appear.
The counts are dependent --- their sum is fixed at $N$ and the validity constraint couples the two exponent bytes within every float --- so the off-diagonal terms of~\eqref{eq:varterms} call for the joint law of $\left(\#(v_0), \#(v_1)\right)$ and not for the product of their individual marginals.
It remains to establish, under the VSTL model, how likely it is for a given value to appear a specified number of times in a block, and for two distinct values to appear specified numbers of times jointly.

\subsection{Distributions of Individual and Joint Byte Counts}

To count the occurrence of a given byte $v$ at an unconstrained position, we define
\begin{equation}\label{eq:Z_U,v}
Z_{U,v}(b_U) =
\begin{cases}
1, & b_U = v, \\
0, & b_U \neq v,
\end{cases}
\qquad b_U \in \mathcal{V}.
\end{equation}
$Z_{U,v}$ without an argument stands for $Z_{U,v}(B_U)$, where $B_U$ is randomly drawn according to~\eqref{eq:pu}. Hence, $Z_{U,v} \sim \mathrm{Bernoulli}(P_{B_U})$.

Likewise, to count the occurrences of a given byte $v$ in an exponent pair, we define
\begin{equation}\label{eq:Z_X,v}
Z_{\mathcal{X},v}(b_0, b_1) =
\begin{cases}
2, & b_0 = b_1 = v, \\
1, & b_0 \neq b_1 \wedge (b_0 = v \vee b_1 = v), \\
0, & b_0 \neq v \wedge b_1 \neq v,
\end{cases}
\qquad (b_0, b_1) \in \mathcal{X}.
\end{equation}
$Z_{\mathcal{X},v}$ without arguments stands for $Z_{\mathcal{X},v}(B_0, B_1)$, where $(B_0, B_1)$ is randomly drawn according to~\eqref{eq:pb01}.

Let $\Lambda_{\mathcal{X},v}(m)$ denote the probability that exactly $m \in \{0, 1, 2\}$ bytes of an exponent pair are equal to $v$.
Its expression follows by inclusion--exclusion on the events $B_0 = v$ and $B_1 = v$:
\begin{equation}\label{eq:lambda_X,v}
\begin{aligned}
\Lambda_{\mathcal{X},v}(m) &= \Pr \left[ Z_{\mathcal{X},v} = m \right], \\
\Lambda_{\mathcal{X},v}(0) &= 1 - P_{B_0}(v) - P_{B_1}(v) + P_{B_{0,1}}(v, v), \\
\Lambda_{\mathcal{X},v}(1) &= P_{B_0}(v) + P_{B_1}(v) - 2 P_{B_{0,1}}(v, v), \\
\Lambda_{\mathcal{X},v}(2) &= P_{B_{0,1}}(v, v).
\end{aligned}
\end{equation}

An additional argument generalizes the law to any number of independent groups.
Specifically, $\Lambda_{\mathcal{X},v}(n, m)$ is the probability that $n$ independent exponent pairs contribute a total of $m$ occurrences of $v$.
The $n$-fold convolution of~\eqref{eq:lambda_X,v} therefore admits the closed form
\begin{equation}\label{eq:lambda_X,v_n}
\Lambda_{\mathcal{X},v}(n, m) = \!\!\! \sum_{\substack{i \ge 0 \\ 2i \le m \\ i \ge m - n}} \!\! \binom{n}{i, m - 2i, n - m + i} \left[ \Lambda_{\mathcal{X},v}(2) \right]^{i} \left[ \Lambda_{\mathcal{X},v}(1) \right]^{m-2i} \left[ \Lambda_{\mathcal{X},v}(0) \right]^{n-m+i}.
\end{equation}

We extend this idea to the unconstrained positions and define $\Lambda_{U,v}(n, m)$ as the probability that $n$ of them contribute a total of $m$ occurrences of $v$.
These contributions are independent and identically distributed Bernoulli variables~\eqref{eq:Z_U,v}, so $\Lambda_{U,v}(n, m)$ admits the binomial form
\begin{equation}\label{eq:lambda_U,v_n}
\Lambda_{U,v}(n, m) = \binom{n}{m} \left[ P_{B_U} \right]^{m} \left[ 1 - P_{B_U} \right]^{n - m} .
\end{equation}

\begin{lemma}\label{lem:singular_count}
For every $v \in \mathcal{V}$ and every block of $T$ records, the count $\#(v)$ is distributed as the convolution of $\Lambda_{U,v}$ over the $26T$ unconstrained positions with $\Lambda_{\mathcal{X},v}$ over the $12T$ exponent pairs.
\end{lemma}
\begin{equation}\label{eq:singular_count_probability}
\begin{aligned}
\Pr[\#(v) = m] &= \left[ \Lambda_{U,v}(26T, \cdot) \ast \Lambda_{\mathcal{X},v}(12T, \cdot) \right](m) \\
&= \sum_{u} \Lambda_{U,v}(26T, u) \Lambda_{\mathcal{X},v}(12T, m - u).
\end{aligned}
\end{equation}

To count the joint occurrences of two distinct bytes $v_0 \neq v_1$, we evaluate~\eqref{eq:Z_U,v} and~\eqref{eq:Z_X,v} at each of the two values, accounting for the constraints a single group imposes on their co-occurrence.

Let $\Lambda_{\mathcal{X},v_0,v_1}(m_0, m_1)$ be the probability that a single exponent pair contributes $m_0$ occurrences of $v_0$ and $m_1$ occurrences of $v_1$.
Each exponent pair holds only two bytes, so it is necessary that $m_0 + m_1 \le 2$.
Thus,
\begin{equation}\label{eq:lambda_X,v0,v1}
\begin{aligned}
\Lambda_{\mathcal{X},v_0,v_1}(m_0, m_1) &= \Pr \left[ Z_{\mathcal{X},v_0} = m_0, Z_{\mathcal{X},v_1} = m_1 \right], \\
\Lambda_{\mathcal{X},v_0,v_1}(2, 0) &= \Lambda_{\mathcal{X},v_0}(2), \\
\Lambda_{\mathcal{X},v_0,v_1}(0, 2) &= \Lambda_{\mathcal{X},v_1}(2), \\
\Lambda_{\mathcal{X},v_0,v_1}(1, 1) &= P_{B_{0,1}}(v_0, v_1) + P_{B_{0,1}}(v_1, v_0), \\
\Lambda_{\mathcal{X},v_0,v_1}(1, 0) &= \Lambda_{\mathcal{X},v_0}(1) - \Lambda_{\mathcal{X},v_0,v_1}(1, 1), \\
\Lambda_{\mathcal{X},v_0,v_1}(0, 1) &= \Lambda_{\mathcal{X},v_1}(1) - \Lambda_{\mathcal{X},v_0,v_1}(1, 1), \\
\Lambda_{\mathcal{X},v_0,v_1}(0, 0) &= 1 - \!\!\!\!\! \sum_{(a, b) \neq (0, 0)} \!\!\!\!\! \Lambda_{\mathcal{X},v_0,v_1}(a, b).
\end{aligned}
\end{equation}

An additional argument again generalizes the law to any number of independent groups.
Specifically, $\Lambda_{\mathcal{X},v_0,v_1}(n, m_0, m_1)$ is the probability that $n$ independent exponent pairs contribute a total of $m_0$ occurrences of $v_0$ and $m_1$ occurrences of $v_1$.
The $n$-fold convolution of~\eqref{eq:lambda_X,v0,v1} therefore admits the closed form
\begin{equation}\label{eq:lambda_X,v0,v1_n}
\begin{gathered}
\Lambda_{\mathcal{X},v_0,v_1}(n, m_0, m_1) = \\
\begin{aligned}
\sum_{\substack{i, j, k \ge 0 \\ 2i + k \le m_0 \\ 2j + k \le m_1 \\ i + j + k \ge m_0 + m_1 - n}} \!\!\!\!\!\!\!\!\!\! &\binom{n}{i, j, k, m_0 - 2i - k, m_1 - 2j - k, n - m_0 - m_1 + i + j + k} \\
&\quad \times \left[ \Lambda_{\mathcal{X},v_0,v_1}(2,0) \right]^{i} \left[ \Lambda_{\mathcal{X},v_0,v_1}(0,2) \right]^{j} \left[ \Lambda_{\mathcal{X},v_0,v_1}(1,1) \right]^{k} \\
&\quad \times \left[ \Lambda_{\mathcal{X},v_0,v_1}(1,0) \right]^{m_0 - 2i - k} \left[ \Lambda_{\mathcal{X},v_0,v_1}(0,1) \right]^{m_1 - 2j - k} \\
&\quad \times \left[ \Lambda_{\mathcal{X},v_0,v_1}(0,0) \right]^{n - m_0 - m_1 + i + j + k}.
\end{aligned}
\end{gathered}
\end{equation}

We extend this idea to the unconstrained positions as well and define $\Lambda_{U,v_0,v_1}(n, m_0, m_1)$ as the probability that $n$ of them contribute a total of $m_0$ occurrences of $v_0$ and $m_1$ occurrences of $v_1$.
A single unconstrained position carries only one value, so $Z_{U,v_0} Z_{U,v_1} = 0$.
The binomial form~\eqref{eq:lambda_U,v_n} therefore becomes trinomial over two equiprobable outcomes~\eqref{eq:pu} and their complement:
\begin{equation}\label{eq:lambda_U,v0,v1_n}
\Lambda_{U,v_0,v_1}(n, m_0, m_1) = \binom{n}{m_0, m_1, n - m_0 - m_1} \left[ P_{B_U} \right]^{m_0 + m_1} \left[ 1 - 2 P_{B_U} \right]^{n - m_0 - m_1}.
\end{equation}

Convolving the two families over their respective group counts gives the two-dimensional counterpart of Lemma~\ref{lem:singular_count}.

\begin{lemma}\label{lem:joint_count}
For all distinct $v_0, v_1 \in \mathcal{V}$ and every block of $T$ records, the pair of counts $\bigl( \#(v_0), \#(v_1) \bigr)$ is jointly distributed as the two-dimensional convolution of $\Lambda_{U,v_0,v_1}$ over the $26T$ unconstrained positions with $\Lambda_{\mathcal{X},v_0,v_1}$ over the $12T$ exponent pairs.
\end{lemma}
\begin{equation}\label{eq:joint_count}
\begin{gathered}
\Pr[\#(v_0) = m_0, \#(v_1) = m_1] = \left[ \Lambda_{U,v_0,v_1}(26T, \cdot, \cdot) \ast \Lambda_{\mathcal{X},v_0,v_1}(12T, \cdot, \cdot) \right](m_0, m_1) \\
= \sum_{u_0, u_1} \Lambda_{U,v_0,v_1}(26T, u_0, u_1) \Lambda_{\mathcal{X},v_0,v_1}(12T, m_0 - u_0, m_1 - u_1).
\end{gathered}
\end{equation}

\subsection{The Exact Moments}

The count laws of Lemmas~\ref{lem:singular_count} and~\ref{lem:joint_count} involve $v$, respectively $(v_0, v_1)$, only through the marginals~\eqref{eq:pu}, \eqref{eq:pb1}, \eqref{eq:pb0} and~\eqref{eq:pb01}.
The coarsest partition of $\mathcal{V}$ on which all four are constant is the set of categories
\begin{equation}\label{eq:cats}
\Gamma = \left\{ \mathcal{V}^{*}_{lo}, \{127\}, \mathcal{V}^{*}_{hi}, \{255\} \right\}.
\end{equation}
Any two bytes from the same category thus contribute equal terms to~\eqref{eq:muterms} and~\eqref{eq:varterms}, as do any two pairs drawn from the same categories.
The sums over individual bytes are therefore rewritten as sums over $\Gamma$, with each category weighted by its cardinality: $127$ for either majority set, $1$ for either minority byte.

For an arbitrary $v \in \gamma$, the two single-count moments of the category are
\begin{equation}\label{eq:Eh}
\begin{aligned}
\theta_\gamma &= E \left[ \varphi \left( \frac{\#(v)}{N} \right) \right] \\
&= \sum_{m} \varphi \left(\frac{m}{N}\right) \Pr \left[ \#(v) = m \right],
\end{aligned}
\end{equation}
\begin{equation}\label{eq:Eh2}
\begin{aligned}
\theta^{(2)}_\gamma &= E \left[ \varphi \left( \frac{\#(v)}{N} \right)^2 \right] \\
&= \sum_{m} \varphi \left(\frac{m}{N}\right)^2 \Pr \left[ \#(v) = m \right].
\end{aligned}
\end{equation}

For arbitrary $v_0 \in \gamma_0$ and $v_1 \in \gamma_1$ with $v_0 \neq v_1$, the joint moment of the two categories is
\begin{equation}\label{eq:Ehh}
\begin{aligned}
\theta_{\gamma_0 \gamma_1} &= E \left[ \varphi \left( \frac{\#(v_0)}{N} \right) \varphi \left( \frac{\#(v_1)}{N} \right) \right] \\
&= \sum_{m_0, m_1} \varphi \left(\frac{m_0}{N}\right) \varphi \left(\frac{m_1}{N}\right) \Pr[\#(v_0) = m_0, \#(v_1) = m_1].
\end{aligned}
\end{equation}

\begin{proposition}\label{prop:moments}
For every triangle count $T$, the block entropy has mean~\eqref{eq:muterms} and variance~\eqref{eq:varterms} given by the cardinality-weighted sums over $\Gamma$
\begin{equation}\label{eq:mean}
\mu^{\mathrm{VSTL}}(T) = \sum_{\gamma \in \Gamma} \lvert \gamma \rvert\, \theta_\gamma,
\end{equation}
\begin{equation}\label{eq:var}
\begin{aligned}
\left[ \sigma^{\mathrm{VSTL}}(T) \right]^2 &= \sum_{\gamma \in \Gamma} \lvert \gamma \rvert \theta^{(2)}_\gamma + \sum_{\substack{\gamma \in \Gamma \\ \lvert \gamma \rvert \ge 2}} \lvert \gamma \rvert \left( \lvert \gamma \rvert - 1 \right) \theta_{\gamma\gamma} \\
&\quad + \sum_{\substack{\gamma_0, \gamma_1 \in \Gamma \\ \gamma_0 \neq \gamma_1}} \lvert \gamma_0 \rvert \lvert \gamma_1 \rvert \theta_{\gamma_0 \gamma_1} - \left[ \mu^{\mathrm{VSTL}}(T) \right]^{2} .
\end{aligned}
\end{equation}
\end{proposition}

\subsection{Consistency with the Uniform Model}

\begin{corollary}\label{cor:uniform}
Under the uniform byte model, with $L=\lvert\mathcal{V}\rvert=256$, Proposition~\ref{prop:moments} reduces to
\begin{equation}\label{eq:wu}
\begin{gathered}
\mu^{\mathrm{uniform}}(T) = L \theta_{\mathcal{V}}, \\
\left[ \sigma^{\mathrm{uniform}}(T) \right]^{2} = L \theta^{(2)}_{\mathcal{V}} + L \left(L - 1 \right) \theta_{\mathcal{V}\mathcal{V}} - \left[ \mu^{\mathrm{uniform}}(T) \right]^{2}.
\end{gathered}
\end{equation}
\end{corollary}

\begin{proof}
Assuming uniformity drops the validity constraint on the exponent pairs.
All byte values are then equivalent, so the partition $\Gamma$ reduces to a set of one category, that category being the alphabet $\mathcal{V}$ itself.
Substituting this partition into~\eqref{eq:mean} and~\eqref{eq:var} yields~\eqref{eq:wu}.

Equation~\eqref{eq:wu} therefore reproduces the mean and variance of~\cite[Sec.~3.3]{local_entropy_2013}.
\end{proof}

\section{The VSTL Global Entropy Test}\label{sec:global}

Throughout this section we assume $N$ grows without bound.
Practical STL files justify this assumption: meshes usually run from thousands of triangles upwards, at fifty bytes per triangle.
Although Proposition~\ref{prop:moments} provides exact reference moments at every file size, their direct evaluation through the nested sums in~\eqref{eq:Eh}--\eqref{eq:Ehh} becomes unnecessarily costly for large files.

What the global measure needs instead is the shape~\cite{adjusted_npcr_uaci_2026} gives to the differential statistics: a reference mean and a reference deviation in closed form, assembled in $O(1)$ arithmetic operations at a fixed numerical precision, and turned into a critical value by a single quantile.
This section derives both VSTL references in second-order asymptotic form, with an explicit $O(N^{-3})$ error bound.

\subsection{The Cross-Entropy Decomposition}

The entropy~\eqref{eq:plugin} of a block is a function of its \emph{empirical byte distribution}
\begin{equation}\label{eq:emp}
Q_S(v) = \frac{\#(v)}{N}, \qquad v \in \mathcal{V},
\end{equation}
so that $\displaystyle H(S) = \sum_{v \in \mathcal{V}} \varphi(Q_S(v))$.

Measured against a fixed probability law, the entropy splits into a cross-entropy and a divergence.
The cross-entropy is linear in $Q_S$, hence a sum of $T$ independent per-record terms, so its mean and its variance follow from the moments of a single record, in closed form and at every file size.
The divergence, by contrast, holds the entire logarithmic nonlinearity in one nonnegative term; its moments cannot be obtained in closed form, and Section~\ref{subsec:secondorder} approximates them by a Taylor expansion.

The law we measure against, and later expand around, is the \emph{record-average marginal}
\begin{equation}\label{eq:pbar}
\bar{P} = \frac{26 P_{B_U} + 12 P_{B_0} + 12 P_{B_1}}{50}.
\end{equation}
It is the average of the fifty per-position marginals of a record~\eqref{eq:pu}--\eqref{eq:pb0} and constant on the four categories~\eqref{eq:cats}:
\begin{equation}\label{eq:pbarvals}
\bar{P}(v) =
\begin{cases}
2129/544000, & v \in \mathcal{V}^{*}_{lo}, \\
1873/544000, & v = 127, \\
1/256, & v \in \mathcal{V}^{*}_{hi}, \\
1869/544000, & v = 255.
\end{cases}
\end{equation}
For an arbitrary $v \in \gamma$, we write $\bar{P}(\gamma) = \bar{P}(v)$ for the common value of the law on the category.
Likewise for $P_{B_U}$, $P_{B_0}$ and $P_{B_1}$.

$\bar{P}(v)$ is the probability that a uniformly random position of a record holds the value $v$, and $\bar{P}$ is also the centre of the empirical distribution $Q_S$: linearity of expectation over the fifty positions of each record gives $E[\#(v)] = N \bar{P}(v)$, and hence $E[Q_S] = \bar{P}$ exactly, at every $T$.
This is precisely what qualifies $\bar{P}$ as the point to expand around: relative to it, the divergence measures how far $Q_S$ strays from the centre of its own fluctuations.

\begin{lemma}\label{lem:crossentropy}
Every block $S$ of $T$ records satisfies the identity
\begin{equation}\label{eq:crossentropy}
H(S) = \sum_{v \in \mathcal{V}} Q_S(v)\lambda \left( \bar{P}(v) \right) - D(Q_S \| \bar{P}) ,
\end{equation}
where
\begin{equation}\label{eq:kl}
D(Q_S \| \bar{P}) = - \sum_{v \in \mathcal{V}} Q_S(v) \lambda \left( \frac{Q_S(v)}{\bar{P}(v)} \right)
\end{equation}
is the Kullback--Leibler divergence in bits~\cite[Sec.~2.3]{cover_thomas_2006}.
\end{lemma}

\begin{proof}
It is convenient to note a property of the surprisal function~\eqref{eq:plugin} first: the quotient rule of the underlying logarithm carries over, so differences of surprisals collapse to ratios of their arguments. For positive $x$ and $y$,
\begin{equation}\label{eq:lambda_quotient_rule}
\lambda(x) - \lambda(y) = -\left( \log_2(x) - \log_2(y) \right) = -\log_2\left(\frac{x}{y}\right) = \lambda \left( \frac{x}{y} \right).
\end{equation}

By~\eqref{eq:pbarvals}, $\forall v \in \mathcal{V}, \bar{P}(v) > 0$; hence $\lambda(\bar{P}(v))$ is finite and the ratio $Q_S(v) / \bar{P}(v)$ is well defined.
Adding and subtracting $\lambda(\bar{P}(v))$ inside each summand of~\eqref{eq:plugin} then gives
\begin{equation}\label{eq:cross_entropy_divergence_decomposition_proof}
\begin{aligned}
H(S)
&= \sum_{v \in \mathcal{V}} \varphi \left( Q_S(v) \right) \\
&= \sum_{v \in \mathcal{V}} Q_S(v) \left[ \lambda \left( Q_S(v) \right) + \lambda \left( \bar{P}(v) \right) - \lambda \left( \bar{P}(v) \right) \right] \\
&= \sum_{v \in \mathcal{V}} Q_S(v) \lambda \left( \bar{P}(v) \right) + \sum_{v \in \mathcal{V}} Q_S(v) \lambda \left( \frac{Q_S(v)}{\bar{P}(v)} \right) \\
&= \sum_{v \in \mathcal{V}} Q_S(v) \lambda \left( \bar{P}(v) \right) - D(Q_S \| \bar{P}).
\end{aligned}
\end{equation}
Terms with $Q_S(v) = 0$ contribute nothing to $H(S)$, since $\varphi(0) = 0$; in the chain above, they likewise vanish because $Q_S(v)$ multiplies every summand and carries its zero through.
The identity is therefore applied only where both of its arguments are positive.
\end{proof}

Write $C_r$ for the \emph{score} of the $r$-th record, namely, the surprisal its fifty bytes carry under $\bar{P}$:
\begin{equation}\label{eq:score}
C_r = \sum_{i=1}^{50} \lambda \left( \bar{P}(S_{r,i}) \right), \qquad r = 1, \ldots, T.
\end{equation}
By Definition~\ref{def:vstl}, the $T$ records are i.i.d., and hence so are their scores.
We denote the score of a generic record by $C$.

Averaged per byte over the block, the scores give back the cross-entropy term of~\eqref{eq:crossentropy}, which we denote by $\bar{C}_S$,
\begin{equation}\label{eq:cbar}
\begin{aligned}
\bar{C}_S &= \sum_{v \in \mathcal{V}} Q_S(v) \lambda \left( \bar{P}(v) \right) \\
&= \frac{1}{N} \sum_{v \in \mathcal{V}} \#(v) \lambda \left( \bar{P}(v) \right) \\
&= \frac{1}{N} \sum_{r=1}^{T} \sum_{i=1}^{50} \lambda \left( \bar{P}(S_{r,i}) \right) \\
&= \frac{1}{N} \sum_{r=1}^{T} C_r.
\end{aligned}
\end{equation}

\begin{lemma}\label{lem:score}
The two moments of the record score are
\begin{equation}\label{eq:Cmom}
\begin{gathered}
E[C] = 26 E \left[ \lambda \left( \bar{P}(B_U) \right) \right] + 12 E \left[ \lambda \left( \bar{P}(B_0) \right) + \lambda \left( \bar{P}(B_1) \right) \right] = 50 H(\bar{P}), \\
\operatorname{Var}[C] = 26 \operatorname{Var} \left[ \lambda \left( \bar{P}(B_U) \right) \right] + 12 \operatorname{Var} \left[ \lambda \left( \bar{P}(B_0) \right) + \lambda \left( \bar{P}(B_1) \right) \right].
\end{gathered}
\end{equation}
Consequently, the moments of the block average score are
\begin{equation}\label{eq:CSmom}
\begin{gathered}
E \left[ \bar{C}_S \right] = \frac{E[C]}{50} = H(\bar{P}), \\
\operatorname{Var} \left[ \bar{C}_S \right] = \frac{\operatorname{Var}[C]}{50 N},
\end{gathered}
\end{equation}
exactly and at every $T$.
\end{lemma}

\begin{proof}
Write $B_{U,r}$ for the multiset of the unconstrained bytes and $B_{\mathcal{X},r}$ for the multiset of its exponent pairs of the $r$-th record.
Regrouping the fifty terms of~\eqref{eq:score} along this split leaves the sum unchanged,
\begin{equation}\label{eq:Cgroups}
C_r = \sum_{b_U \in B_{U,r}} \lambda \left( \bar{P}(b_U) \right) + \sum_{(b_0, b_1) \in B_{\mathcal{X},r}} \left( \lambda \left( \bar{P}(b_0) \right) + \lambda \left( \bar{P}(b_1) \right) \right).
\end{equation}
Under (A2) the $38$ summands are mutually independent.

\emph{Mean.}
Linearity of expectation in~\eqref{eq:Cgroups} gives
\begin{equation}\label{eq:ECr}
\begin{aligned}
E[C] &= 26 E \left[ \lambda \left( \bar{P}(B_U) \right) \right] + 12 E \left[ \lambda \left( \bar{P}(B_0) \right) \right] + 12 E \left[ \lambda \left( \bar{P}(B_1) \right) \right] \\
&= 26 \sum_{v \in \mathcal{V}} P_{B_U}(v) \lambda \left( \bar{P}(v) \right) + 12 \sum_{v \in \mathcal{V}} P_{B_0}(v) \lambda \left( \bar{P}(v) \right) + 12 \sum_{v \in \mathcal{V}} P_{B_1}(v) \lambda \left( \bar{P}(v) \right) \\
&= \sum_{v \in \mathcal{V}} \left[ 26 P_{B_U} + 12 P_{B_0} + 12 P_{B_1} \right](v) \lambda \left( \bar{P}(v) \right) \\
&= 50 \sum_{v \in \mathcal{V}} \varphi \left( \bar{P}(v) \right) = 50 H(\bar{P}).
\end{aligned}
\end{equation}

The same contribution can be computed while keeping the exponent pair $(B_0,B_1)$ intact and evaluating it under its joint law $P_{B_{0,1}}$.
This approach uses the fact that~\eqref{eq:pb0} and~\eqref{eq:pb1} are the marginals of~\eqref{eq:pb01}, which Table~\ref{tab:marginals} verifies.
By extending the sum to $\mathcal{V}^2$ --- with no change in value, since $P_{B_{0,1}}$ vanishes off $\mathcal{X}$ --- and rearranging the resulting terms, we obtain the same expression:
\begin{equation}\label{eq:pairsplitsum}
\begin{gathered}
\sum_{(v_0, v_1) \in \mathcal{X}} P_{B_{0,1}}(v_0, v_1) \left( \lambda \left( \bar{P}(v_0) \right) + \lambda \left( \bar{P}(v_1) \right) \right) \\
= \sum_{v_0, v_1 \in \mathcal{V}} P_{B_{0,1}}(v_0, v_1) \lambda \left( \bar{P}(v_0) \right) + \sum_{v_0, v_1 \in \mathcal{V}} P_{B_{0,1}}(v_0, v_1) \lambda \left( \bar{P}(v_1) \right) \\
= \sum_{v_0 \in \mathcal{V}} \lambda \left( \bar{P}(v_0) \right) \sum_{v_1 \in \mathcal{V}} P_{B_{0,1}}(v_0, v_1) + \sum_{v_1 \in \mathcal{V}} \lambda \left( \bar{P}(v_1) \right) \sum_{v_0 \in \mathcal{V}} P_{B_{0,1}}(v_0, v_1) \\
= \sum_{v \in \mathcal{V}} P_{B_0}(v) \lambda \left( \bar{P}(v) \right) + \sum_{v \in \mathcal{V}} P_{B_1}(v) \lambda \left( \bar{P}(v) \right).
\end{gathered}
\end{equation}
Thus, for the mean, evaluating the pair jointly gives the same result as evaluating its two components through their marginal laws.

\begin{table}
\caption{
Verification that~\eqref{eq:pb0} and~\eqref{eq:pb1} are the marginals of~\eqref{eq:pb01}.
Each interior cell gives the probability mass contributed by the indicated partner categories.
Zero indicates that~\eqref{eq:X} excludes every pair in that block.
The entries in each row therefore sum directly to the stated marginal probability.
}

\label{tab:marginals}
\centering
\footnotesize
\begin{subtable}{\textwidth}
\centering
\subcaption{$\displaystyle P_{B_0}(v_0) = \sum_{v_1 \in \mathcal{V}} P_{B_{0,1}}(v_0, v_1)$}
\begin{tabular}{lccc}
\toprule
$v_0$ & $v_1 \in \mathcal{V}_{lo}$ & $v_1 \in \mathcal{V}_{hi}$ & $P_{B_0}(v_0)$ \\
\midrule
$\mathcal{V}^{*}_{lo} \cup \mathcal{V}^{*}_{hi}$ & $\dfrac{128}{65280}$ & $\dfrac{128}{65280}$ & $\dfrac{1}{255}$ \\ [8pt]
$\{127, 255\}$ & $\dfrac{128}{65280}$ & $0$ & $\dfrac{1}{510}$ \\
\bottomrule
\end{tabular}
\end{subtable}

\bigskip

\begin{subtable}{\textwidth}
\centering
\subcaption{$\displaystyle P_{B_1}(v_1) = \sum_{v_0 \in \mathcal{V}} P_{B_{0,1}}(v_0, v_1)$}
\begin{tabular}{lccc}
\toprule
$v_1$ & $v_0 \in \mathcal{V}^{*}_{lo} \cup \mathcal{V}^{*}_{hi}$ & $v_0 \in \{127, 255\}$ & $P_{B_1}(v_1)$ \\
\midrule
$\mathcal{V}_{lo}$ & $\dfrac{254}{65280}$ & $\dfrac{2}{65280}$ & $\dfrac{1}{255}$ \\ [8pt]
$\mathcal{V}_{hi}$ & $\dfrac{254}{65280}$ & $0$ & $\dfrac{127}{32640}$ \\
\bottomrule
\end{tabular}
\end{subtable}
\end{table}

Since the $T$ record scores in~\eqref{eq:cbar} have the common mean $E[C]$, the expected value of the block-average score is
\begin{equation}\label{eq:CbarMean}
E \left[ \bar{C}_S \right] = \frac{T E[C]}{N} = \frac{E[C]}{50} = H(\bar{P}).
\end{equation}

\emph{Variance.}
The $38$ groups of~\eqref{eq:Cgroups} are the finest partition of a record into independent parts, and variances add over them and over no coarser split.
Because validity couples the two bytes of every exponent pair, they must be treated jointly, with their contribution given by the variance of $\lambda(\bar{P}(B_0)) + \lambda(\bar{P}(B_1))$, which does not equal the sum of the corresponding single-byte variances, as shown in~\eqref{eq:paircov} below.
By the per-position marginals of Section~\ref{subsec:layout}, the $26$ unconstrained groups and the $12$ exponent pairs are identically distributed under~\eqref{eq:pu} and~\eqref{eq:pb01}, respectively, giving
\begin{equation}\label{eq:varCform}
\operatorname{Var}[C] = 26 \operatorname{Var} \left[ \lambda \left( \bar{P}(B_U) \right) \right] + 12 \operatorname{Var} \left[ \lambda \left( \bar{P}(B_0) \right) + \lambda \left( \bar{P}(B_1) \right) \right] .
\end{equation}
Expanding each of the two terms leaves
\begin{equation}\label{eq:varmoments}
\begin{gathered}
\operatorname{Var} \left[ \lambda \left( \bar{P}(B_U) \right) \right] = E \left[ \lambda \left( \bar{P}(B_U) \right)^2 \right] - \left( E \left[ \lambda \left( \bar{P}(B_U) \right) \right] \right)^2 , \\
\begin{aligned}
\operatorname{Var} \left[ \lambda \left( \bar{P}(B_0) \right) + \lambda \left( \bar{P}(B_1) \right) \right] &= E \left[ \left( \lambda \left( \bar{P}(B_0) \right) + \lambda \left( \bar{P}(B_1) \right) \right)^{2} \right] \\
&\quad - \left( E \left[ \lambda \left( \bar{P}(B_0) \right) + \lambda \left( \bar{P}(B_1) \right) \right] \right)^2 .
\end{aligned}
\end{gathered}
\end{equation}

Unlike the mean, the variance cannot be expressed in terms of $\bar{P}$ alone.
The squared means in~\eqref{eq:varmoments} are not linear in the three class marginals, so~\eqref{eq:pbar} cannot fold them into $\bar{P}$ as it did in~\eqref{eq:ECr}.
Moreover, expanding the exponent pair's second moment introduces the cross term $E \left[ \lambda \left( \bar{P}(B_0) \right) \lambda \left( \bar{P}(B_1) \right) \right]$, which depends on the joint law~\eqref{eq:pb01}, not on its marginals.
The evaluation of the variance therefore requires computing the per-class moments and the moments of the exponent pair under its joint law.

\begin{table}
\caption{The surprisal a byte carries under $\bar{P}$, by category~\eqref{eq:cats}, from the probabilities~\eqref{eq:pbarvals}.}
\label{tab:catscores}
\centering
\footnotesize
\resizebox{\textwidth}{!}{%
\begin{tabular}{lcccc}
\toprule
 & $\mathcal{V}^{*}_{lo}$ & $\{127\}$ & $\mathcal{V}^{*}_{hi}$ & $\{255\}$ \\
\midrule
$\lambda \left( \bar{P}(\gamma) \right)$ & $7.997286891460$ & $8.182111941669$ & $8.000000000000$ & $8.185196272167$ \\
\bottomrule
\end{tabular}
}
\end{table}

By~\eqref{eq:pu}--\eqref{eq:pb0} and~\eqref{eq:pbarvals}, the three marginals and the surprisal under $\bar{P}$ are constant on the categories~\eqref{eq:cats}, the values of the latter being reported in Table~\ref{tab:catscores}.
Hence, for $B_{\circ} \in \{ B_U, B_0, B_1 \}$ and $m = 1, 2$,
\begin{equation}\label{eq:byterule}
\begin{aligned}
E \left[ \lambda \left( \bar{P}(B_{\circ}) \right)^{m} \right] &= \sum_{v \in \mathcal{V}} P_{B_{\circ}}(v) \lambda \left( \bar{P}(v) \right)^{m} \\
&= \sum_{\gamma \in \Gamma} \lvert \gamma \rvert P_{B_{\circ}}(\gamma) \lambda \left( \bar{P}(\gamma) \right)^{m}.
\end{aligned}
\end{equation}
The category partitioning thus reduces each expectation to four summands since all $\lvert \gamma \rvert$ bytes of a given category contribute the same term.

For the exponent pair the weights must instead come from the joint law $P_{B_{0,1}}$.
It is uniform on $\mathcal{X}$ by~\eqref{eq:pb01}, and the surprisal depends on $(v_0, v_1)$ only through their categories, so, to compute a pair moment, only the number of admissible pairs in each cell of $\Gamma \times \Gamma$ is needed:
\begin{equation}\label{eq:cells}
n_{\gamma_0 \gamma_1} =
\begin{array}{c|cccc}
\gamma_0 \big\backslash \gamma_1 & \mathcal{V}^{*}_{lo} & \{127\} & \mathcal{V}^{*}_{hi} & \{255\} \\
\hline
\mathcal{V}^{*}_{lo} & 16129 & 127 & 16129 & 127 \\
\{127\} & 127 & 1 & 0 & 0 \\
\mathcal{V}^{*}_{hi} & 16129 & 127 & 16129 & 127 \\
\{255\} & 127 & 1 & 0 & 0
\end{array}
\end{equation}
Every cell holds all $\lvert \gamma_0 \rvert \lvert \gamma_1 \rvert$ pairs of $\gamma_0 \times \gamma_1$, except the four that the IEEE-754 validity constraint discards.

A moment of the exponent pair therefore reduces to a sixteen-term sum over $\Gamma \times \Gamma$. For $m = 1,2$,
\begin{equation}\label{eq:pairrule}
\begin{aligned}
E \left[ \left( \lambda \left( \bar{P}(B_0) \right) + \lambda \left( \bar{P}(B_1) \right) \right)^{m} \right]  &= \!\!\! \sum_{(v_0, v_1) \in \mathcal{X}} \!\!\! P_{B_{0,1}}(v_0, v_1) \left( \lambda \left( \bar{P}(v_0) \right) + \lambda \left( \bar{P}(v_1) \right) \right)^{m} \\
&= \frac{1}{65280} \sum_{\gamma_0, \gamma_1 \in \Gamma} n_{\gamma_0 \gamma_1} \left( \lambda \left( \bar{P}(\gamma_0) \right) + \lambda \left( \bar{P}(\gamma_1) \right) \right)^{m},
\end{aligned}
\end{equation}
The common denominator factors out, while $n_{\gamma_0 \gamma_1}$ gives the multiplicity of each category pair.

\begin{table}
\caption{First two moments and variances of the byte and exponent-pair surprisal variables under $\bar{P}$.}
\label{tab:moments}
\centering
\footnotesize
\resizebox{\textwidth}{!}{%
\begin{tabular}{lcccc}
\toprule
& $\lambda \left( \bar{P}(B_U) \right)$ & $\lambda \left( \bar{P}(B_0) \right)$ & $\lambda \left( \bar{P}(B_1) \right)$ & $\lambda \left( \bar{P}(B_0) \right) + \lambda \left( \bar{P}(B_1) \right)$ \\
\midrule
$E \left[ \bullet \right]$ & $8.000088841521$ & $7.999368977735$ & $8.000083516013$ & $15.999452493748$ \\
$E \left[ \bullet^{2} \right]$ & $64.001688641127$ & $63.990039589053$ & $64.001603429977$ & $255.982880850774$ \\
$\operatorname{Var} \left[ \bullet \right]$ & $0.000267168905$ & $0.000135547109$ & $0.000267166789$ & $0.000400751074$ \\
\bottomrule
\end{tabular}
}
\end{table}

Table~\ref{tab:moments} reports the first and second moments of the unconstrained-byte and exponent-pair surprisals, required by the two variances in~\eqref{eq:varmoments} and evaluated using~\eqref{eq:byterule} and~\eqref{eq:pairrule}.
Substitution into~\eqref{eq:varCform} then gives

\begin{equation}\label{eq:varCnum}
\begin{gathered}
\operatorname{Var}[C] \approx 0.011755404409.
\end{gathered}
\end{equation}

Since the $T$ record scores in~\eqref{eq:cbar} are independent and have the common variance $\operatorname{Var}[C]$, the variance of the block-average score is
\begin{equation}\label{eq:CbarVar}
\operatorname{Var} \left[ \bar{C}_S \right] = \frac{T \operatorname{Var}[C]}{N^{2}} = \frac{\operatorname{Var}[C]}{50 N}.
\end{equation}
Together with~\eqref{eq:CbarMean}, this establishes~\eqref{eq:CSmom}.

\emph{Covariance within the exponent pair.}
To quantify the dependence within the exponent pair, we compute the covariance of its two byte-surprisal contributions:
\begin{equation}\label{eq:paircov}
\begin{gathered}
\begin{aligned}
\operatorname{Cov} \left[ \lambda \left( \bar{P}(B_0) \right), \lambda \left( \bar{P}(B_1) \right) \right] &= \frac{\operatorname{Var} \left[ \lambda \left( \bar{P}(B_0) \right) + \lambda \left( \bar{P}(B_1) \right) \right]}{2} \\
&\quad - \frac{\operatorname{Var} \left[ \lambda \left( \bar{P}(B_0) \right) \right] + \operatorname{Var} \left[ \lambda \left( \bar{P}(B_1) \right) \right]}{2}
\end{aligned} \\
\operatorname{Cov} \left[ \lambda \left( \bar{P}(B_0) \right), \lambda \left( \bar{P}(B_1) \right) \right] \approx -0.000000981412.
\end{gathered}
\end{equation}
Although the covariance is small, omitting it would overestimate $\operatorname{Var}[C]$ by approximately $0.000023553883$, or $0.2\%$.
\end{proof}

\subsection{The Entropy Ceiling}\label{subsec:ceiling}

Section~\ref{sec:prelim} defines $H_\infty$ as the limit the block entropy attains as the string grows and gives $H_\infty = 8$ bits for independent uniform bytes.
For the VSTL source, if $Q_S = \bar{P}$, then~\eqref{eq:crossentropy} reduces to $H(S) = H(\bar{P})$, with $D(Q_S \,\|\, \bar{P}) = 0$.
Lemma~\ref{lem:score} shows that $H(\bar{P})$ is also the exact mean of the cross-entropy term for every $T$, making it the candidate value of $H_\infty^{\mathrm{VSTL}}$.
The following proposition establishes it.

\begin{proposition}\label{prop:ceiling}
Let
\begin{equation}\label{eq:ceiling}
H_\infty^{\mathrm{VSTL}} = H(\bar{P}) \approx 7.99991479609.
\end{equation}
Then $H(S) \to H_\infty^{\mathrm{VSTL}}$ almost surely as $T \to \infty$.
Moreover, at every $T$,
\begin{equation}\label{eq:ceilingid}
\mu^{\mathrm{VSTL}}(T) = H_\infty^{\mathrm{VSTL}} - E \left[ D(Q_S \| \bar{P}) \right] < H_\infty^{\mathrm{VSTL}}.
\end{equation}
The deficit from $H_\infty$, the entropy ceiling under the uniform-byte model, is
\begin{equation}\label{eq:deficit}
\Delta = H_\infty - H_\infty^{\mathrm{VSTL}} \approx 0.00008520391 \text{ bits}.
\end{equation}
\end{proposition}

\begin{proof}
For each $v \in \mathcal{V}$, its counts in the $T$ records are i.i.d.\ by Definition~\ref{def:vstl}, bounded, and have common mean $50 \bar{P}(v)$ by~\eqref{eq:pbar}.
The strong law of large numbers~\cite[Sec.~7.5]{grimmett_stirzaker_2001} gives $Q_S(v) \to \bar{P}(v)$ almost surely for each $v$; since $\mathcal{V}$ is finite, these coordinatewise convergences hold simultaneously, and hence $Q_S \to \bar{P}$ almost surely.
The continuity of $H$ on the probability simplex then gives $H(S) = H(Q_S) \to H(\bar{P}) = H_\infty^{\mathrm{VSTL}}$ almost surely.

Taking expectations in~\eqref{eq:crossentropy} and using $E[\bar{C}_S] = H(\bar{P})$ from Lemma~\ref{lem:score} gives~\eqref{eq:ceilingid}.
Since the divergence is nonnegative and vanishes only at $Q_S = \bar{P}$, its expectation is strictly positive, proving the inequality.
Indeed, for every finite $T$, the empirical byte distribution is not fixed at $\bar{P}$ and therefore differs from it with positive probability.

Finally, $\bar{P}$ is nonuniform, so $H(\bar{P}) = H_\infty^{\mathrm{VSTL}} < H_\infty$, and~\eqref{eq:deficit} follows.
\end{proof}

One distinction remains: $H_\infty^{\mathrm{VSTL}}$ is the limiting entropy of the empirical byte distribution, not the entropy rate of the VSTL source.
Each record consists of the $26$ uniform bytes and $12$ exponent pairs in~\eqref{eq:Cgroups}, so its entropy is
\begin{equation}\label{eq:recordentropy}
26 \lambda \left( \frac{1}{\lvert \mathcal{V} \rvert} \right) + 12 \lambda \left( \frac{1}{\lvert \mathcal{X} \rvert} \right) = 12 \lambda \left( \frac{1}{\lvert \mathcal{F} \rvert} \right) + 2 \lambda \left( \frac{1}{\lvert \mathcal{V} \rvert} \right).
\end{equation}
The source entropy rate is therefore
\begin{equation}\label{eq:sourcerate}
h^{\mathrm{VSTL}} := \frac{26 H(P_{B_U}) + 12 H(P_{B_{0,1}})}{50}.
\end{equation}
The marginal and joint entropies are reported in Table~\ref{tab:sourceentropies}.

\begin{table}
\caption{Entropies of the three byte marginals and the joint law of an exponent pair.}
\label{tab:sourceentropies}
\centering
\footnotesize
\resizebox{\textwidth}{!}{%
\begin{tabular}{lcccc}
\toprule
& $P_{B_U}$ & $P_{B_0}$ & $P_{B_1}$ & $P_{B_{0,1}}$ \\
\midrule
$H(\bullet)$ & $8$ & $7.998275005486$ & $7.999988906584$ & $15.994353436859$ \\
\bottomrule
\end{tabular}
}
\end{table}

Two standard inequalities then give
\begin{equation}\label{eq:rate}
H_\infty^{\mathrm{VSTL}} \ge \frac{26 H(P_{B_U}) + 12 \left( H(P_{B_0}) + H(P_{B_1}) \right)}{50} \ge h^{\mathrm{VSTL}}.
\end{equation}
The first inequality is the concavity of $H$~\cite[Thm.~2.7.3]{cover_thomas_2006}, applied to $\bar{P}$; the difference is approximately $0.000331457193$ bits per byte.
The second inequality is the independence bound~\cite[Thm.~2.6.6]{cover_thomas_2006}.
For each exponent pair, the difference between the sum of the marginal entropies and the joint entropy is the mutual information
\begin{equation}\label{eq:pairmi}
I(B_0; B_1) = H(P_{B_0}) + H(P_{B_1}) - H(P_{B_{0,1}}) \approx 0.003910475212 \text{ bits}.
\end{equation}
Accumulated over all exponent pairs and normalized per byte, this gives the difference of approximately $0.000938514051$ bits per byte between the two sides of the second inequality.
Overall, the total gap is $0.001269971244$ bits per byte, about fifteen times $\Delta$.

These two differences explain why the limiting entropy of the empirical byte distribution differs from the source entropy rate.
Reading~\eqref{eq:rate} from right to left, the independence bound first replaces the joint law $P_{B_{0,1}}$ in $h^{\mathrm{VSTL}}$ with its two marginals.
Concavity then replaces the position-dependent marginals with their average $\bar{P}$.
Accordingly, $h^{\mathrm{VSTL}}$ reflects both the dependence within exponent pairs and the differences among positional marginals, whereas $H_\infty^{\mathrm{VSTL}}$ is the entropy of their pooled byte law $\bar{P}$.

Because the global entropy statistic operates at this aggregate level, the analysis below uses $H_\infty^{\mathrm{VSTL}}$ as its reference value rather than $h^{\mathrm{VSTL}}$.

\subsection{Second-Order Expansions of the Reference Values}\label{subsec:secondorder}

A Taylor expansion of $H(S)$ about $\bar{P}$ relates fluctuations of the empirical byte distribution to the reference mean and variance.
The decomposition established in Lemma~\ref{lem:crossentropy} separates a linear cross-entropy score, whose moments are given by Lemma~\ref{lem:score}, from a divergence, which is expanded coordinatewise with a remainder that quantifies the approximation error.
Each coordinate of the whole-file empirical distribution averages contributions from independent, identically distributed records, so the required expectations and covariances can be calculated from the joint moments for a single record.

The difference between the empirical byte distribution and its mean recurs throughout the expansion.
For brevity, this fluctuation at any $v \in \mathcal{V}$ is abbreviated as
\begin{equation}\label{eq:delta_S,v}
\delta_{S,v} = Q_S(v) - \bar{P}(v).
\end{equation}

For a given $K \ge 2$ and a function $f$ with $K+1$ continuous derivatives between $a$ and $x$, Taylor's formula with an integral remainder is
\begin{equation}\label{eq:taylorgeneric}
f(x) = \sum_{k=0}^{K} \frac{f^{(k)}(a)}{k!} (x - a)^{k} + \int_{a}^{x} \frac{(x - t)^{K}}{K!} f^{(K+1)}(t) dt ,
\end{equation}
This identity follows from the fundamental theorem of calculus after $K$ integrations by parts~\cite[Ch.~7]{apostol_1967}.

Entropy is separable across byte values: each summand involves a single coordinate, so its Taylor expansion requires only the univariate derivatives of $\varphi$ given by
\begin{equation}\label{eq:phider}
\begin{gathered}
\varphi'(x) = \lambda \left( x \right) - \frac{1}{\ln(2)}, \\
\varphi^{(k)}(x) = \frac{(-1)^{k-1} (k-2)!}{\ln(2) x^{k-1}}, \quad k \ge 2.
\end{gathered}
\end{equation}
The formula for $\varphi^{(k)}$ is used in the classical entropy expansion of Basharin~\cite[Eq.~(5)]{basharin_1959}.

Applying~\eqref{eq:taylorgeneric} to each coordinate, with $a=\bar{P}(v)$ and $x=Q_S(v)$, gives
\begin{equation}\label{eq:htaylor}
\begin{aligned}
H(S) = {}& \sum_{v \in \mathcal{V}} \varphi \left( \bar{P}(v) \right) + \sum_{v \in \mathcal{V}} \varphi' \left( \bar{P}(v) \right) \delta_{S,v} \\
&+ \sum_{k=2}^{K} \sum_{v \in \mathcal{V}}  \frac{\delta_{S,v}^{k} \varphi^{(k)} \left( \bar{P}(v) \right)}{k!} + \sum_{v \in \mathcal{V}} R_{K,v},
\end{aligned}
\end{equation}
where $R_{K,v}$ is the coordinatewise remainder
\begin{equation}\label{eq:remdef}
R_{K,v} = \int_{\bar{P}(v)}^{Q_S(v)} \frac{\left( Q_S(v) - t \right)^{K}}{K!} \varphi^{(K+1)}(t) dt.
\end{equation}
As established earlier, the expansion point is the exact mean of the empirical byte distribution, so the deviations $\delta_{S,v}$ are centred.
Furthermore, since every $\bar{P}(v)$ is positive, all derivatives in~\eqref{eq:htaylor} are finite at the expansion point.

The linear term in~\eqref{eq:htaylor} can now be simplified by substituting the expression for $\varphi'(x)$ from~\eqref{eq:phider}:
\begin{equation}\label{eq:linterm}
\begin{aligned}
\sum_{v \in \mathcal{V}} \delta_{S,v} \varphi' \left( \bar{P}(v) \right) &= \sum_{v \in \mathcal{V}} Q_S(v) \lambda \left( \bar{P}(v) \right) - \sum_{v \in \mathcal{V}} \varphi \left( \bar{P}(v) \right) \\
&\quad  - \frac{1}{\ln(2)} \left( \sum_{v \in \mathcal{V}} Q_S(v) - \sum_{v \in \mathcal{V}} \bar{P}(v) \right) \\
&= \bar{C}_S - H_\infty^{\mathrm{VSTL}}.
\end{aligned}
\end{equation}
Histogram normalization gives $\displaystyle \sum_{v \in \mathcal{V}} Q_S(v) = \sum_{v \in \mathcal{V}} \bar{P}(v) = 1$; therefore, the constant term in $\varphi'(x)$ makes no contribution.

Substituting~\eqref{eq:linterm} into~\eqref{eq:htaylor} and comparing the result with Lemma~\ref{lem:crossentropy} isolates the divergence:
\begin{equation}\label{eq:dtaylor}
- D(Q_S \| \bar{P}) = \sum_{k=2}^{K} \frac{(-1)^{k-1}}{\ln(2)k(k-1)} \sum_{v \in \mathcal{V}} \frac{\delta_{S,v}^{k}}{\bar{P}(v)^{k-1}} + \sum_{v \in \mathcal{V}} R_{K,v}.
\end{equation}
The equation remains separable across byte values and expresses the nonlinear part of the entropy in powers of the per-coordinate deviations, beginning with their squares.

To determine an appropriate value for $K$ so that the reference mean and variance are approximated with sufficient accuracy, we must first analyze the error left after truncating the expansion at that degree.
For each coordinate, this error is $R_{K,v}$.
In the following calculations, its magnitude is shown to be bounded by a constant multiple of $\lvert \delta_{S,v} \rvert^{K+1}$.

Substituting the formula for $\varphi^{(K+1)}$ from~\eqref{eq:phider} into~\eqref{eq:remdef} gives
\begin{equation}\label{eq:remratio}
\begin{aligned}
R_{K,v} &= \int_{\bar{P}(v)}^{Q_S(v)} \frac{\left(Q_S(v) - t \right)^K}{K!} \cdot \frac{(-1)^K (K-1)!}{\ln(2)t^K} dt \\
&= \frac{(-1)^K}{\ln(2)K} \int_{\bar{P}(v)}^{Q_S(v)} \left( \frac{Q_S(v) - t}{t} \right)^K dt,
\end{aligned}
\end{equation}
This expression is rewritten by parametrizing the line segment from $\bar{P}(v)$ to $Q_S(v)$.
This substitution allows us to factor out a power of the deviation $\delta_{S,v}$, leaving an integral that can be bounded uniformly.
The parameter runs from $u=0$ at the expansion point to $u=1$ at the empirical probability, giving
\begin{equation}\label{eq:subs}
\begin{gathered}
t = \bar{P}(v) + u \delta_{S,v}, \qquad u \in [0, 1], \\
t \big|_{u=0} = \bar{P}(v), \qquad t \big|_{u=1} = Q_S(v), \\
dt = \delta_{S,v} du,
\end{gathered}
\end{equation}
which transforms~\eqref{eq:remratio} into
\begin{equation}\label{eq:remsub}
\begin{aligned}
R_{K,v} &= \frac{(-1)^K}{\ln(2)K} \int_{0}^{1} \left( \frac{\left( 1 - u \right) \delta_{S,v}}{\bar{P}(v) + u \delta_{S,v}} \right)^{K} \delta_{S,v} du \\
&= \frac{(-1)^K \delta_{S,v}^{K+1}}{\ln(2)K} \int_{0}^{1} \left( \frac{1 - u}{\left( 1 - u \right) \bar{P}(v) + u Q_S(v)} \right)^{K} du.
\end{aligned}
\end{equation}

When $Q_S(v)=0$ and $u=1$, the ratio in~\eqref{eq:remsub} is undefined, but the integral can still be defined as an improper integral by approaching the upper integration endpoint from below~\cite[Sec.~10.23]{apostol_1967}.
The resulting limit exists because, for $Q_S(v)=0$ and $u \in [0,1)$, the ratio is constant at $1/\bar{P}(v)$.
Assigning this same value at $u=1$ therefore gives a continuous extension whose ordinary integral agrees with that limit.

The preceding argument establishes the following inequalities:
\begin{equation}\label{eq:kernel}
\begin{gathered}
(1 - u) \bar{P}(v) \le (1 - u) \bar{P}(v) + u Q_S(v), \\
0 \le \frac{1 - u}{(1 - u) \bar{P}(v) + u Q_S(v)} \le \frac{1}{\bar{P}(v)}, \\
0 \le \int_{0}^{1} \left( \frac{1 - u}{(1 - u) \bar{P}(v) + u Q_S(v)} \right)^{K} du \le \int_{0}^{1} \frac{1}{\bar{P}(v)^{K}} du = \frac{1}{\bar{P}(v)^{K}}.
\end{gathered}
\end{equation}
Combining this result with~\eqref{eq:remsub} gives the error bound
\begin{equation}\label{eq:rembound}
\lvert R_{K,v} \rvert \le \frac{\lvert \delta_{S,v} \rvert^{K+1}}{\ln(2) K \bar{P}(v)^{K}} .
\end{equation}

\subsubsection{The Reference Mean}\label{subsubsec:asymean}

To proceed with the asymptotic analysis, we first require moments of the histogram deviations.
Fix a byte $v \in \mathcal{V}$ and denote its count in the $r$-th record by $X_{v,r}\in\{0,\ldots,50\}$.
By Definition~\ref{def:vstl}, these counts are i.i.d.\ across the $T$ records, and their generic representative is denoted by $X_v$.
Using the multisets $B_{U,r}$ and $B_{\mathcal{X},r}$ introduced in~\eqref{eq:Cgroups} and the counting functions $Z_{U,v}$ and $Z_{\mathcal{X},v}$ defined in~\eqref{eq:Z_U,v} and~\eqref{eq:Z_X,v}, we obtain
\begin{equation}\label{eq:XvsQ}
\begin{gathered}
X_{v,r} = \sum_{b_U \in B_{U,r}} Z_{U,v}(b_U) + \sum_{(b_0, b_1) \in B_{\mathcal{X},r}} Z_{\mathcal{X},v}(b_0, b_1), \\
Q_S(v) = \frac{\#(v)}{N} = \frac{1}{N} \sum_{r=1}^{T} X_{v,r}, \\
E \left[ X_v \right] = 50 \bar{P}(v).
\end{gathered}
\end{equation}
Thus, although $Q_S(v)$ is computed over the entire block, it is a normalized sum of independent record-level counts.

The corresponding centred record count is therefore
\begin{equation}\label{eq:Xtilde}
\tilde{X}_v := X_v - E \left[ X_v \right].
\end{equation}
By construction, $E[\tilde{X}_v]=0$~\cite[Def.~(3.3.5)]{grimmett_stirzaker_2001}, but for every higher order the central moment follows from the binomial expansion~\cite[Sec.~5-4]{papoulis_pillai_2002}
\begin{equation}\label{eq:centralbinom}
E \left[ \tilde{X}_v^{k} \right] = \sum_{i=0}^{k} (-1)^{k-i} \binom{k}{i} E \left[ X_v \right]^{k-i} E \left[ X_v^{i} \right].
\end{equation}

\begin{proposition}\label{prop:asymean}
As $T \to \infty$, the VSTL reference mean admits the following second-order expansion:
\begin{equation}\label{eq:asymean}
\begin{aligned}
\mu^{\mathrm{VSTL}}(T) = {}& H_\infty^{\mathrm{VSTL}} - \frac{1}{100 \ln(2) N} \sum_{v \in \mathcal{V}} \frac{E[\tilde{X}_v^{2}]}{\bar{P}(v)} \\
- &\frac{1}{100 \ln(2) N^{2}} \sum_{v \in \mathcal{V}} \left[ \frac{E[\tilde{X}_v^{2}]^{2}}{100 \bar{P}(v)^{3}} - \frac{E[\tilde{X}_v^{3}]}{3 \bar{P}(v)^{2}} \right] + O(N^{-3}).
\end{aligned}
\end{equation}
\end{proposition}

\begin{proof}
By Proposition~\ref{prop:ceiling} and~\eqref{eq:ceilingid}, it suffices to expand $E \left[ D(Q_S \| \bar{P}) \right]$.

Equations~\eqref{eq:XvsQ} and~\eqref{eq:Xtilde} show that each histogram deviation is a normalized sum of centred i.i.d.\ record counts:
\begin{equation}\label{eq:devsum}
\begin{aligned}
\delta_{S,v} &= \frac{1}{N} \sum_{r=1}^{T} X_{v,r} - \bar{P}(v) \\
&= \frac{1}{N} \sum_{r=1}^{T} \left( X_{v,r} - 50 \bar{P}(v) \right) \\
&= \frac{1}{N} \sum_{r=1}^{T} \tilde{X}_{v,r}.
\end{aligned}
\end{equation}
Consequently, the moments of the histogram deviations can be expressed through moments of the centred record counts.

Raising the final identity in~\eqref{eq:devsum} to the $k$-th power yields a sum over $k$-tuples $(r_1, \ldots, r_k)$ of not necessarily distinct record indices, where each summand is a product of the form $\tilde{X}_{v,r_1} \cdots \tilde{X}_{v,r_k}$.
Independence factorizes the expectation of each term, with each represented record contributing a centred moment of order equal to its multiplicity.
Since $E[\tilde{X}_v]=0$, a term vanishes unless each of its constituent records has multiplicity at least two.
For $k \ge 2$,
\begin{equation}\label{eq:devexplicit}
\begin{aligned}
E \left[ \delta_{S,v}^{k} \right] &= \frac{1}{N^{k}} \sum_{r_1 = 1}^{T} \cdots \sum_{r_k = 1}^{T} E \left[ \tilde{X}_{v,r_1} \cdots \tilde{X}_{v,r_k} \right] \\
&= \frac{1}{N^{k}}\displaystyle \sum_{\ell = 1}^{\lfloor k/2 \rfloor} \binom{T}{\ell} \sum_{\substack{a_1, \ldots, a_\ell \ge 2 \\ a_1 + \cdots + a_\ell = k}} \binom{k}{a_1, \ldots, a_\ell} \prod_{i=1}^{\ell} E \left[ \tilde{X}_v^{a_i} \right].
\end{aligned}
\end{equation}
The integer $\ell$ is the number of distinct records represented in a surviving $k$-tuple, and the binomial coefficient $\binom{T}{\ell}$ counts the ways to select these $\ell$ records from the block.
The integer $a_i$ is the multiplicity of the $i$-th selected record, and the multinomial coefficient $\binom{k}{a_1,\ldots,a_\ell}$ counts the corresponding arrangements.

Since every $a_i\ge2$ and $a_1+\cdots+a_\ell=k$, at most $\lfloor k/2\rfloor$ distinct records can be represented in an arrangement of the $k$ factors.
Consequently, the upper limit on $\ell$ in~\eqref{eq:devexplicit} also determines the asymptotic order
\begin{equation}\label{eq:mean_prob_diff_asymptotics}
E \left[ \delta_{S,v}^{k} \right] = O \left( N^{- \lceil k / 2 \rceil} \right).
\end{equation}

Therefore, the deviation moments contributing through $O(N^{-2})$ in the Taylor expansion are
\begin{equation}\label{eq:dmoments}
\begin{aligned}
E \left[ \delta_{S,v}^{2} \right] = {}& \frac{E \left[ \tilde{X}_v^{2} \right]}{50 N}, \\
E \left[ \delta_{S,v}^{3} \right] = {}& \frac{E \left[ \tilde{X}_v^{3} \right]}{50 N^2}, \\
E \left[ \delta_{S,v}^{4} \right] = {}& \frac{3 E \left[ \tilde{X}_v^{2} \right]^{2}}{2500 N^2} + \frac{E \left[ \tilde{X}_v^{4} \right] - 3 E \left[ \tilde{X}_v^{2} \right]^{2} }{50 N^3}. \\
\end{aligned}
\end{equation}

Equations~\eqref{eq:rembound} and~\eqref{eq:mean_prob_diff_asymptotics} now provide everything needed to choose an appropriate value for $K$ so that the remainder bound has the required asymptotic order.
The first equation involves an absolute moment and the second, a signed one; for odd $K+1$, the two are conceptually different: positive and negative deviations can cancel in the signed moment, whereas both contribute positively to the absolute moment.
However, the two coincide when $K+1$ is even, so we choose the smallest such power whose moment is $O(N^{-3})$, namely $K+1=6$, and hence fix $K=5$.
For this choice, the resulting bound is
\begin{equation}\label{eq:remorder}
\begin{gathered}
\sum_{v \in \mathcal{V}} \Bigl\lvert E \left[ R_{5,v} \right] \Bigr\rvert \le \sum_{v \in \mathcal{V}} \frac{E \left[ \delta_{S,v}^{6} \right]}{5 \ln(2) \bar{P}(v)^{5}} \\
\sum_{v \in \mathcal{V}} \Bigl\lvert E \left[ R_{5,v} \right] \Bigr\rvert = O(N^{-3}).
\end{gathered}
\end{equation}

Substituting~\eqref{eq:dmoments} into~\eqref{eq:dtaylor} determines the leading terms, while $O(N^{-3})$ absorbs all remaining terms and, by~\eqref{eq:remorder}, the remainder.
Equation~\eqref{eq:ceilingid} then gives~\eqref{eq:asymean}.
\end{proof}

To compute the leading terms in~\eqref{eq:asymean}, the moments of $\tilde{X}_v$ must be evaluated.
For the two group counts in~\eqref{eq:XvsQ}, define the centred variables
\begin{equation}\label{eq:groupcopies}
\begin{gathered}
\tilde{Z}_{U,v} := Z_{U,v} - E \left[ Z_{U,v} \right], \\
\tilde{Z}_{\mathcal{X},v} := Z_{\mathcal{X},v} - E \left[ Z_{\mathcal{X},v} \right].
\end{gathered}
\end{equation}
Taking expectations in~\eqref{eq:XvsQ} and subtracting the resulting identity from the original decomposition writes $\tilde{X}_v$ as a sum of mutually independent centred variables.
Squaring this sum and taking expectations leaves only the diagonal terms: every cross term vanishes because independence factorizes it into a product of zero means.
Cubing the sum likewise leaves only the pure-cube terms: every mixed term contains at least one centred variable to the first power and therefore has zero expectation.
Thus,
\begin{equation}\label{eq:varsigma}
\begin{gathered}
E \left[ \tilde{X}_v^{2} \right] = 26 E \left[ \tilde{Z}_{U,v}^{2} \right] + 12 E \left[ \tilde{Z}_{\mathcal{X},v}^{2} \right], \\
E \left[ \tilde{X}_v^{3} \right] = 26 E \left[ \tilde{Z}_{U,v}^{3} \right] + 12 E \left[ \tilde{Z}_{\mathcal{X},v}^{3} \right].
\end{gathered}
\end{equation}
Equivalently, these identities express the additivity of cumulants over independent summands~\cite[Sec.~5-5]{papoulis_pillai_2002}.

For the Bernoulli count $Z_{U,v}$, its indicator definition in~\eqref{eq:Z_U,v} implies the standard idempotence property
\begin{equation}\label{eq:E_Z_U,v^k}
\begin{gathered}
Z_{U,v}^k = Z_{U,v}, \\
E \left[ Z_{U,v}^k \right] = P_{B_U},
\end{gathered}
\qquad k \ge 1.
\end{equation}
Applying the same binomial expansion as in~\eqref{eq:centralbinom} to $\tilde{Z}_{U,v}$ for $k = 2,3$ therefore yields
\begin{equation}\label{eq:meanZ_U}
\begin{aligned}
E \left[ \tilde{Z}_{U,v}^{2} \right] &= E \left[ Z_{U,v}^{2} \right] - E \left[ Z_{U,v} \right]^{2} \\
&= P_{B_U} (1 - P_{B_U}), \\
E \left[ \tilde{Z}_{U,v}^{3} \right] &= E \left[ Z_{U,v}^{3} \right] - 3 E \left[ Z_{U,v} \right] E \left[ Z_{U,v}^{2} \right] + 2 E \left[ Z_{U,v} \right]^{3} \\
&= P_{B_U} (1 - P_{B_U})(1 - 2 P_{B_U}).
\end{aligned}
\end{equation}

The moment of order $k$ of $Z_{\mathcal{X},v}$ under its law $\Lambda_{\mathcal{X},v}$~\eqref{eq:lambda_X,v} is
\begin{equation}
E \left[ Z_{\mathcal{X},v}^{k} \right] = \Lambda_{\mathcal{X},v}(1) + 2^k \Lambda_{\mathcal{X},v}(2), \qquad k \ge 1.
\end{equation}
The same binomial expansion as in~\eqref{eq:centralbinom}, now applied to $\tilde{Z}_{\mathcal{X},v}$ for $k = 2,3$, gives
\begin{equation}\label{eq:meanZ_X}
\begin{aligned}
E \left[ \tilde{Z}_{\mathcal{X},v}^{2} \right] &= E \left[ Z_{\mathcal{X},v}^{2} \right] - E \left[ Z_{\mathcal{X},v} \right]^{2} \\
&= \Lambda_{\mathcal{X},v}(1) + 4 \Lambda_{\mathcal{X},v}(2) - \left( \Lambda_{\mathcal{X},v}(1) + 2 \Lambda_{\mathcal{X},v}(2) \right)^{2}, \\
E \left[ \tilde{Z}_{\mathcal{X},v}^{3} \right] &= E \left[ Z_{\mathcal{X},v}^{3} \right] - 3 E \left[ Z_{\mathcal{X},v} \right] E \left[ Z_{\mathcal{X},v}^{2} \right] + 2 E \left[ Z_{\mathcal{X},v} \right]^{3} \\
&= \Lambda_{\mathcal{X},v}(1) + 8 \Lambda_{\mathcal{X},v}(2) - 3 \left( \Lambda_{\mathcal{X},v}(1) + 2 \Lambda_{\mathcal{X},v}(2) \right) \left( \Lambda_{\mathcal{X},v}(1) + 4 \Lambda_{\mathcal{X},v}(2) \right) \\
& \quad + 2 \left( \Lambda_{\mathcal{X},v}(1) + 2 \Lambda_{\mathcal{X},v}(2) \right)^{3}.
\end{aligned}
\end{equation}

Substituting~\eqref{eq:varsigma}, \eqref{eq:meanZ_U}, and~\eqref{eq:meanZ_X} into~\eqref{eq:asymean} gives
\begin{equation}\label{eq:meanconstants}
\begin{gathered}
\frac{1}{100 \ln(2)} \sum_{v \in \mathcal{V}} \frac{E[\tilde{X}_v^{2}]}{\bar{P}(v)} \approx 183.943323730404 \\
\frac{1}{100 \ln(2)} \sum_{v \in \mathcal{V}} \left[ \frac{E[\tilde{X}_v^{2}]^{2}}{100 \bar{P}(v)^{3}} - \frac{E[\tilde{X}_v^{3}]}{3 \bar{P}(v)^{2}} \right] \approx 7879.931052184461
\end{gathered}
\end{equation}
The numerical expansion of the reference mean therefore is
\begin{equation}\label{eq:meanassembled}
\boxed{
\begin{aligned}
\mu^{\mathrm{VSTL}}(T) &= 7.99991479609 - \frac{183.943323730404}{N} \\
&\quad - \frac{7879.931052184461}{N^2} + O \left( N^{-3} \right)
\end{aligned}
}.
\end{equation}

\subsubsection{The Reference Variance}\label{subsubsec:asyvar}

From the cross-entropy decomposition and the established expansion, the reference mean has been determined by taking expectations; taking variances instead expresses the reference variance as
\begin{equation}\label{eq:varsplit}
\left[ \sigma^{\mathrm{VSTL}}(T) \right]^2 = \operatorname{Var} \left[ \bar{C}_S \right] - 2 \operatorname{Cov} \left[ \bar{C}_S, D(Q_S \| \bar{P}) \right] + \operatorname{Var}[D(Q_S \| \bar{P})].
\end{equation}

The first term in~\eqref{eq:varsplit} follows from Lemma~\ref{lem:score}.
The covariance term is obtained from~\eqref{eq:dtaylor} with $K=5$ by linearity of covariance, thereby separating the contributions of the polynomial part of the Taylor expansion from its remainder:
\begin{equation}\label{eq:covexpand}
\begin{gathered}
\operatorname{Cov} \left[ \bar{C}_S, D(Q_S \| \bar{P}) \right] = - \operatorname{Cov} \left[ \bar{C}_S, \sum_{k=2}^{5} \frac{(-1)^{k-1}}{\ln(2)k(k-1)} \sum_{v \in \mathcal{V}} \frac{\delta_{S,v}^{k}}{\bar{P}(v)^{k-1}} + \sum_{v \in \mathcal{V}} R_{5,v} \right] \\
= - \sum_{k=2}^{5} \frac{(-1)^{k-1}}{\ln(2)k(k-1)} \sum_{v \in \mathcal{V}} \frac{\operatorname{Cov} \left[ \bar{C}_S, \delta_{S,v}^{k} \right]}{\bar{P}(v)^{k-1}} - \sum_{v \in \mathcal{V}} \operatorname{Cov} \left[ \bar{C}_S, R_{5,v} \right].
\end{gathered}
\end{equation}
Analogously, the variance of the divergence is obtained from the same expansion by bilinearity of covariance:
\begin{equation}\label{eq:varDexpand}
\begin{gathered}
\operatorname{Var}[D(Q_S \| \bar{P})] = \operatorname{Var} \left[ \sum_{k=2}^{5} \frac{(-1)^{k-1}}{\ln(2)k(k-1)} \sum_{v \in \mathcal{V}} \frac{\delta_{S,v}^{k}}{\bar{P}(v)^{k-1}} + \sum_{v \in \mathcal{V}} R_{5,v} \right] \\
= \frac{1}{\ln(2)^{2}} \sum_{k_0, k_1 = 2}^{5} \frac{(-1)^{k_0-1}}{k_0 (k_0 - 1)} \frac{(-1)^{k_1-1}}{k_1 (k_1 - 1)} \sum_{v_0, v_1 \in \mathcal{V}} \frac{\operatorname{Cov} \left[ \delta_{S,v_0}^{k_0}, \delta_{S,v_1}^{k_1} \right]}{\bar{P}(v_0)^{k_0 - 1} \bar{P}(v_1)^{k_1 - 1}} \\
+ \frac{2}{\ln(2)} \sum_{k = 2}^{5} \frac{(-1)^{k-1}}{k(k-1)} \sum_{v_0, v_1 \in \mathcal{V}} \frac{\operatorname{Cov} \left[ \delta_{S,v_0}^{k}, R_{5,v_1} \right]}{\bar{P}(v_0)^{k-1}} + \operatorname{Var} \left[ \sum_{v \in \mathcal{V}} R_{5,v} \right].
\end{gathered}
\end{equation}

Denote the centred score of the $r$-th record by
\begin{equation}\label{eq:Ctilde}
\tilde{C}_r := C_r - E \left[ C \right].
\end{equation}
By Definition~\ref{def:vstl}, these centred scores are i.i.d.\ across the $T$ records; write $\tilde{C}$ for a generic representative.

Taking expectations in the final line of~\eqref{eq:cbar} and subtracting the resulting identity from the original gives
\begin{equation}\label{eq:Cbarcentred}
\bar{C}_S - E \left[ \bar{C}_S \right] = \frac{1}{N} \sum_{r=1}^{T} \tilde{C}_r.
\end{equation}

\begin{proposition}\label{prop:asyvar}
As $T \to \infty$, the VSTL reference variance admits the following second-order expansion:
\begin{equation}\label{eq:asyvar}
\begin{aligned}
\left[ \sigma^{\mathrm{VSTL}}(T) \right]^{2} = {}& \frac{\operatorname{Var} \left[ C \right]}{50 N} - \frac{1}{50 \ln(2) N^{2}} \sum_{v \in \mathcal{V}} \frac{E \left[ \tilde{C} \tilde{X}_v^{2} \right]}{\bar{P}(v)} \\
& + \frac{1}{2500 \ln(2) N^{2}} \sum_{v \in \mathcal{V}}\frac{E \left[ \tilde{C} \tilde{X}_v \right] E \left[ \tilde{X}_v^{2} \right]}{\bar{P}(v)^{2}} \\
& + \frac{1}{5000 \ln(2)^2 N^2} \sum_{v_0, v_1 \in \mathcal{V}} \frac{E \left[ \tilde{X}_{v_0} \tilde{X}_{v_1} \right]^{2}}{\bar{P}(v_0) \bar{P}(v_1)} + O \left( N^{-3} \right).
\end{aligned}
\end{equation}
\end{proposition}

\begin{proof}
We evaluate in turn all covariances and variances appearing in~\eqref{eq:covexpand} and~\eqref{eq:varDexpand}, retaining every contribution through order $N^{-2}$.

We begin with the covariance between the score fluctuation and powers of the deviations of the empirical byte distribution.
By~\eqref{eq:Cbarcentred}, the score fluctuation is an average of centred record scores and therefore has expectation zero.
Consequently, the covariance reduces to the product moment
\begin{equation}\label{eq:covraw}
\operatorname{Cov} \left[ \bar{C}_S, \delta_{S,v}^{k} \right] = E \left[ \delta_{S,v}^{k} \left( \bar{C}_S - E[\bar{C}_S] \right) \right].
\end{equation}

Substituting~\eqref{eq:devsum} and~\eqref{eq:Cbarcentred} into~\eqref{eq:covraw} yields a sum over $(k+1)$-tuples $(r_0,r_1,\ldots,r_k)$ of not necessarily distinct record indices, where each summand is a product of the form $\tilde{C}_{r_0}\tilde{X}_{v,r_1}\cdots\tilde{X}_{v,r_k}$.
Independence factorizes the expectation of each term; a record of multiplicity $m$ contributes $E[\tilde{C}\tilde{X}_v^{m-1}]$ when it is indexed by $r_0$, and $E[\tilde{X}_v^m]$ otherwise.
Since $E[\tilde{C}]=E[\tilde{X}_v]=0$, a term vanishes unless $r_0$ coincides with at least one of $r_1,\ldots,r_k$ and every other represented record occurs at least twice.
For $k \ge 2$,
\begin{equation}\label{eq:covexplicit}
\begin{gathered}
E \left[ \delta_{S,v}^{k} \left( \bar{C}_S - E[\bar{C}_S] \right) \right] = \frac{1}{N^{k+1}} \sum_{r_0 = 1}^{T} \sum_{r_1 = 1}^{T} \cdots \sum_{r_k = 1}^{T} E \left[ \tilde{C}_{r_0} \tilde{X}_{v,r_1} \cdots \tilde{X}_{v,r_k} \right] \\
= \frac{1}{50 N^k} \sum_{\ell = 0}^{\lfloor (k-1)/2 \rfloor} \binom{T-1}{\ell} \sum_{\substack{a_1, \ldots, a_\ell \ge 2, \; b \ge 1 \\ a_1 + \cdots + a_\ell + b = k}} \binom{k}{a_1, \ldots, a_\ell, b } E \left[ \tilde{C} \tilde{X}_v^{b} \right] \prod_{i=1}^{\ell} E \left[ \tilde{X}_v^{a_i} \right].
\end{gathered}
\end{equation}

The index $r_0$ identifies the record contributing $\tilde{C}$.
The integer $\ell$ is the number of additional distinct records represented among the $k$ deviation factors, and the binomial coefficient $\binom{T-1}{\ell}$ counts the ways to select them from the remaining records.
The integer $b$ is the multiplicity of $r_0$ among those factors, while $a_i$ is the multiplicity of the $i$-th additional record; the multinomial coefficient $\binom{k}{a_1,\ldots,a_\ell,b}$ counts the corresponding arrangements.

Since $b\ge1$ and every $a_i\ge2$, at most $\lfloor(k-1)/2\rfloor$ additional records can be represented among the $k$ deviation factors.
Consequently, the upper limit on $\ell$ in~\eqref{eq:covexplicit} determines the asymptotic order
\begin{equation}\label{eq:cov_prob_diff_asymptotics}
\operatorname{Cov} \left[ \bar{C}_S, \delta_{S,v}^{k} \right] = O \left( N^{- \lfloor k / 2 \rfloor - 1} \right).
\end{equation}

Therefore, the covariance terms contributing through $O(N^{-2})$ in~\eqref{eq:covexpand} are
\begin{equation}\label{eq:covterms}
\begin{aligned}
\operatorname{Cov} \left[ \bar{C}_S, \delta_{S,v}^{2} \right] &= \frac{E \left[ \tilde{C} \tilde{X}_v^{2} \right]}{50 N^2}, \\
\operatorname{Cov} \left[ \bar{C}_S, \delta_{S,v}^{3} \right] &= \frac{3 E \left[ \tilde{C} \tilde{X}_v \right] E \left[ \tilde{X}_v^{2} \right]}{2500 N^2} + \frac{E \left[ \tilde{C} \tilde{X}_v^{3} \right] - 3 E \left[ \tilde{C} \tilde{X}_v \right] E \left[ \tilde{X}_v^{2} \right]}{50 N^3}. \\
\end{aligned}
\end{equation}

We next consider the covariance between powers of the deviations at two byte values $v_0$ and $v_1$.
The covariance is obtained by subtracting the product of the two marginal moments from their joint moment:
\begin{equation}\label{eq:pair_cov}
\operatorname{Cov} \left[ \delta_{S,v_0}^{k_0}, \delta_{S,v_1}^{k_1} \right] = E \left[ \delta_{S,v_0}^{k_0} \delta_{S,v_1}^{k_1} \right] - E \left[ \delta_{S,v_0}^{k_0} \right] E \left[ \delta_{S,v_1}^{k_1} \right].
\end{equation}
It remains only to evaluate the joint moment, since equations~\eqref{eq:devexplicit}--\eqref{eq:dmoments} already determine the product of the marginal moments in~\eqref{eq:pair_cov}.

Substituting~\eqref{eq:devsum} for both deviations in $E \left[ \delta_{S,v_0}^{k_0} \delta_{S,v_1}^{k_1} \right]$ yields a sum over $(k_0+k_1)$-tuples of not necessarily distinct record indices, where each summand contains $k_0$ factors associated with $v_0$ and $k_1$ factors associated with $v_1$.
Independence factorizes the expectation of each term; a record represented $a$ times among the factors associated with $v_0$ and $b$ times among those associated with $v_1$ contributes the joint centred moment $E[\tilde{X}_{v_0}^{a}\tilde{X}_{v_1}^{b}]$.
Since $E[\tilde{X}_{v_0}]=E[\tilde{X}_{v_1}]=0$, a term vanishes unless every represented record carries at least two factors in total.
For $k_0, k_1 \ge 2$, we have
\begin{equation}\label{eq:varexplicit}
\begin{gathered}
E \left[ \delta_{S,v_0}^{k_0} \delta_{S,v_1}^{k_1} \right] = \frac{1}{N^{k_0+k_1}} \sum_{r_1 = 1}^{T} \cdots \sum_{r_{k_0+k_1} = 1}^{T} E \left[ \prod_{i=1}^{k_0} \tilde{X}_{v_0,r_i} \prod_{i=1}^{k_1} \tilde{X}_{v_1,r_{k_0+i}} \right] \\
= \frac{1}{N^{k_0+k_1}} \sum_{\ell = 1}^{\lfloor (k_0+k_1)/2 \rfloor} \binom{T}{\ell} \sum_{\substack{a_1 + \cdots + a_\ell = k_0 \\ b_1 + \cdots + b_\ell = k_1 \\ a_i, b_i \ge 0 \\ a_i + b_i \ge 2}} \binom{k_0}{a_1, \ldots, a_\ell} \binom{k_1}{b_1, \ldots, b_\ell} \prod_{i=1}^{\ell} E \left[ \tilde{X}_{v_0}^{a_i} \tilde{X}_{v_1}^{b_i} \right].
\end{gathered}
\end{equation}

The integer $\ell$ is the number of distinct records represented in a surviving $(k_0+k_1)$-tuple, and the binomial coefficient $\binom{T}{\ell}$ counts the ways to select them from the block.
The integers $a_i$ and $b_i$ are the multiplicities of the $i$-th selected record among the factors associated with $v_0$ and $v_1$, respectively; the two multinomial coefficients count the corresponding arrangements within the two groups.

Since every $a_i+b_i\ge2$ and their sum over all represented records is $k_0+k_1$, at most $\lfloor(k_0+k_1)/2\rfloor$ distinct records can appear.
Consequently, the upper limit on $\ell$ in~\eqref{eq:varexplicit} determines the asymptotic order
\begin{equation}\label{eq:pair_prob_diff_asymptotics}
E \left[ \delta_{S,v_0}^{k_0} \delta_{S,v_1}^{k_1} \right] = O \left( N^{- \lceil (k_0 + k_1) / 2 \rceil} \right)
\end{equation}

It follows from~\eqref{eq:mean_prob_diff_asymptotics} and~\eqref{eq:pair_prob_diff_asymptotics} that the product of the marginal moments is of no greater asymptotic order than the joint moment. Their difference~\eqref{eq:pair_cov} therefore obeys the same bound:
\begin{equation}\label{eq:pair_cov_asymptotics}
\operatorname{Cov} \left[ \delta_{S,v_0}^{k_0}, \delta_{S,v_1}^{k_1} \right] = O \left( N^{- \lceil (k_0 + k_1) / 2 \rceil} \right)
\end{equation}

The only joint moment required through $O(N^{-2})$ is
\begin{equation}\label{eq:varchi}
\begin{aligned}
E \left[ \delta_{S,v_0}^{2} \delta_{S,v_1}^{2} \right] &= \frac{E \left[ \tilde{X}_{v_0}^{2} \right] E \left[ \tilde{X}_{v_1}^{2} \right] + 2 E \left[ \tilde{X}_{v_0} \tilde{X}_{v_1} \right]^{2}}{2500 N^{2}} \\
&\quad + \frac{E \left[ \tilde{X}_{v_0}^{2} \tilde{X}_{v_1}^{2} \right] - E \left[ \tilde{X}_{v_0}^{2} \right] E \left[ \tilde{X}_{v_1}^{2} \right] - 2 E \left[ \tilde{X}_{v_0} \tilde{X}_{v_1} \right]^{2}}{50 N^{3}}. \\
\end{aligned}
\end{equation}

All remaining terms involve the Taylor remainder, whose variance provides the common bound needed for their analysis.
Since $R_{5,v}$ is square-integrable, we may use the standard identity $\operatorname{Var}[Y]=E[Y^{2}]-E[Y]^{2}$, which implies $\operatorname{Var}[Y] \le E[Y^{2}]$.
Squaring~\eqref{eq:rembound} at $K=5$, taking expectations, and applying~\eqref{eq:mean_prob_diff_asymptotics} with $k=12$ therefore determines the asymptotic order of $\operatorname{Var}[R_{5,v}]$:
\begin{equation}\label{eq:remsq}
\begin{gathered}
\operatorname{Var} \left[ R_{5,v} \right] \le E \left[ R_{5,v}^{2} \right] \le \frac{E \left[ \delta_{S,v}^{12} \right]}{25 \ln(2)^{2} \bar{P}(v)^{10}}, \\
\operatorname{Var} \left[ R_{5,v} \right] = O(N^{-6}).
\end{gathered}
\end{equation}

The next term to be analyzed is the covariance between the score fluctuation and the remainder.
The Cauchy--Schwarz inequality bounds the absolute value of the covariance between two random variables by the square root of the product of their variances~\cite[App.~10B]{papoulis_pillai_2002}.
Lemma~\ref{lem:score} shows that $\operatorname{Var}[\bar{C}_S]=O(N^{-1})$; combining this with the bound for $R_{5,v}$ in~\eqref{eq:remsq} gives
\begin{equation}\label{eq:cov_cs_r5}
\begin{gathered}
\Bigl\lvert \operatorname{Cov} \left[ \bar{C}_S, R_{5,v} \right] \Bigr\rvert \le \sqrt{\operatorname{Var} \left[ \bar{C}_S \right] \operatorname{Var} \left[ R_{5,v} \right]} \le \sqrt{\operatorname{Var} \left[ \bar{C}_S \right] E \left[ R_{5,v}^{2} \right]} \\
\Bigl\lvert \operatorname{Cov} \left[ \bar{C}_S, R_{5,v} \right] \Bigr\rvert = O \left( N^{-7/2} \right)
\end{gathered}
\end{equation}

Likewise, the same Cauchy--Schwarz inequality is applied to $\operatorname{Cov}[\delta_{S,v_0}^{k},R_{5,v_1}]$, with $v_0$ indexing the deviation and $v_1$ the remainder~\cite[App.~10B]{papoulis_pillai_2002}.
Equation~\eqref{eq:mean_prob_diff_asymptotics}, together with the general variance bound established above, shows that $\operatorname{Var} \left[ \delta_{S,v_0}^{k} \right] \le E \left[ \delta_{S,v_0}^{2k} \right] = O(N^{-k})$; combining this with the bound for $R_{5,v_1}$ in~\eqref{eq:remsq} gives
\begin{equation}\label{eq:cov_prob_diff_r5}
\begin{gathered}
\Bigl\lvert \operatorname{Cov} \left[ \delta_{S,v_0}^{k}, R_{5,v_1} \right] \Bigr\rvert \le \sqrt{\operatorname{Var} \left[ \delta_{S,v_0}^{k} \right] \operatorname{Var} \left[ R_{5,v_1} \right]} \\
\Bigl\lvert \operatorname{Cov} \left[ \delta_{S,v_0}^{k}, R_{5,v_1} \right] \Bigr\rvert = O \left( N^{- k / 2 - 3} \right)
\end{gathered}
\end{equation}

The final remainder contribution in~\eqref{eq:varDexpand} is the variance of the sum of $R_{5,v}$ over all byte values.
The finite-dimensional Cauchy--Schwarz inequality~\cite[Thm.~1.41]{apostol_1967} then bounds the squared sum before the expectation is taken:
\begin{equation}\label{eq:var_sum_r5}
\begin{gathered}
\operatorname{Var} \left[ \sum_{v \in \mathcal{V}} R_{5,v} \right] \le E \left[ \left( \sum_{v \in \mathcal{V}} R_{5,v} \right)^2 \right] \le \lvert \mathcal{V} \rvert \sum_{v \in \mathcal{V}} E \left[ R_{5,v}^{2} \right] \\
\operatorname{Var} \left[ \sum_{v \in \mathcal{V}} R_{5,v} \right] = O \left( N^{-6} \right) .
\end{gathered}
\end{equation}

Substituting the results from~\eqref{eq:covterms}, \eqref{eq:pair_cov}, \eqref{eq:dmoments}, and~\eqref{eq:varchi} into the expansions~\eqref{eq:covexpand} and~\eqref{eq:varDexpand} determines all contributions through order $N^{-2}$.
Equations~\eqref{eq:cov_prob_diff_asymptotics}, \eqref{eq:pair_cov_asymptotics}, and~\eqref{eq:remsq}--\eqref{eq:var_sum_r5} show that every remaining contribution is $O(N^{-3})$ or smaller.
Equation~\eqref{eq:varsplit}, together with Lemma~\ref{lem:score}, then gives~\eqref{eq:asyvar}.
\end{proof}

To evaluate the mixed moments in~\eqref{eq:asyvar}, we first decompose the centred score $\tilde{C}_r$ into contributions from the unconstrained bytes and exponent pairs of the $r$-th record, by analogy with~\eqref{eq:Cgroups} and~\eqref{eq:XvsQ}:
\begin{equation}\label{eq:groupcentred}
\begin{gathered}
\tilde{C}_r = \sum_{b_U \in B_{U,r}} \tilde{C}_U(b_U) + \sum_{(b_0,b_1) \in B_{\mathcal{X},r}} \tilde{C}_{\mathcal{X}}(b_0,b_1), \\
\tilde{C}_U(b_U) := \lambda \left( \bar{P}(b_U) \right) - E \left[ \lambda \left( \bar{P}(B_U) \right) \right], \\
\tilde{C}_{\mathcal{X}}(b_0, b_1) := \lambda \left( \bar{P}(b_0) \right) + \lambda \left( \bar{P}(b_1) \right) - E \left[ \lambda \left( \bar{P}(B_0) \right) + \lambda \left( \bar{P}(B_1) \right) \right].
\end{gathered}
\end{equation}
$\tilde{C}_U$ without an argument stands for $\tilde{C}_U(B_U)$, where $B_U$ is randomly drawn according to~\eqref{eq:pu}.
Likewise, $\tilde{C}_{\mathcal{X}}$ without arguments stands for $\tilde{C}_{\mathcal{X}}(B_0, B_1)$, where $(B_0, B_1)$ is randomly drawn according to~\eqref{eq:pb01}.
As a consequence,
\begin{equation}
E \left[ \tilde{C} \right] = 26 E \left[ \tilde{C}_U \right] + 12 E \left[ \tilde{C}_{\mathcal{X}} \right].
\end{equation}
Together with the centred count decomposition from~\eqref{eq:groupcopies}, this gives:
\begin{equation}\label{eq:mixedadd}
\begin{aligned}
E \left[ \tilde{X}_{v_0} \tilde{X}_{v_1} \right] &= 26 E \left[ \tilde{Z}_{U,v_0} \tilde{Z}_{U,v_1} \right] + 12 E \left[ \tilde{Z}_{\mathcal{X},v_0} \tilde{Z}_{\mathcal{X},v_1} \right], \\
E \left[ \tilde{C} \tilde{X}_v \right] &= 26 E \left[ \tilde{C}_U \tilde{Z}_{U,v} \right] + 12 E \left[ \tilde{C}_{\mathcal{X}} \tilde{Z}_{\mathcal{X},v} \right], \\
E \left[ \tilde{C} \tilde{X}_v^{2} \right] &= 26 E \left[ \tilde{C}_U \tilde{Z}_{U,v}^{2} \right] + 12 E \left[ \tilde{C}_{\mathcal{X}} \tilde{Z}_{\mathcal{X},v}^{2} \right].
\end{aligned}
\end{equation}
Independence factorizes moments combining contributions from unconstrained bytes with contributions from exponent pairs, and centring makes them vanish: for example, $E \left[ \tilde{C}_U\tilde{Z}_{\mathcal{X},v} \right] = E \left[ \tilde{C}_U \right] E \left[ \tilde{Z}_{\mathcal{X},v} \right] = 0$.

For a single unconstrained byte, the product of the occurrence indicators is given by
\begin{equation}\label{eq:Z_U_identities}
\left[ Z_{U,v_0} Z_{U,v_1} \right](b_U) =
\begin{cases}
1, & v_0 = v_1 = b_U, \\
0, & v_0 \neq b_U \vee v_1 \neq b_U,
\end{cases}
\qquad b_U \in \mathcal{V}.
\end{equation}
The two indicators coincide only when $v_0=v_1=v$, in which case $Z_{U,v_0}Z_{U,v_1}=Z_{U,v}$, consistently with~\eqref{eq:E_Z_U,v^k}; when $v_0\neq v_1$, they are mutually exclusive and their product is zero.
Taking expectations gives
\begin{equation}\label{eq:E_Z_U_v0_Z_U_v1}
E \left[ Z_{U,v_0} Z_{U,v_1} \right] =
\begin{cases}
P_{B_U}, & v_0 = v_1, \\
0, & v_0 \neq v_1.
\end{cases}
\end{equation}
It follows that the centred product has expectation
\begin{equation}\label{eq:sigZ_U}
\begin{aligned}
E \left[ \tilde{Z}_{U,v_0} \tilde{Z}_{U,v_1} \right]
&= E \left[ Z_{U,v_0} Z_{U,v_1} \right] - E \left[ Z_{U,v_0} \right] E \left[ Z_{U,v_1} \right] \\
&= E \left[ Z_{U,v_0} Z_{U,v_1} \right] - P_{B_U}^{2}.
\end{aligned}
\end{equation}
When evaluating the double sum over $(v_0,v_1)$ in~\eqref{eq:asyvar}, the two cases from~\eqref{eq:E_Z_U_v0_Z_U_v1}, implicitly inherited by~\eqref{eq:sigZ_U}, apply as follows: the first on the diagonal, for equal target byte values, and the second off the diagonal, for distinct target byte values.

For an exponent pair, the product of the occurrence counts is given by
\begin{equation}\label{eq:Z_X_identities}
\left[ Z_{\mathcal{X},v_0} Z_{\mathcal{X},v_1} \right](b_0, b_1) =
\begin{cases}
4, & b_0 = b_1 = v_0 = v_1, \\
1, & b_0 \neq b_1 \wedge (b_0 = v_0 \vee b_1 = v_0) \wedge (b_0 = v_1 \vee b_1 = v_1), \\
0, & \left( b_0 \neq v_0 \wedge b_1 \neq v_0 \right) \vee \left( b_0 \neq v_1 \wedge b_1 \neq v_1 \right).
\end{cases}
\end{equation}
Taking expectations gives
\begin{equation}\label{eq:E_Z_X_v0_Z_X_v1}
E \left[ Z_{\mathcal{X},v_0} Z_{\mathcal{X},v_1} \right] =
\begin{cases}
E \left[ Z_{\mathcal{X},v}^{2} \right] = \Lambda_{\mathcal{X},v}(1) + 4 \Lambda_{\mathcal{X},v}(2), & v_0 = v_1 = v, \\
\Lambda_{\mathcal{X},v_0,v_1}(1,1), & v_0 \neq v_1.
\end{cases}
\end{equation}
The corresponding centred moment is then
\begin{equation}\label{eq:sigZ_X}
\begin{aligned}
E \left[ \tilde{Z}_{\mathcal{X},v_0} \tilde{Z}_{\mathcal{X},v_1} \right]
&= E \left[ Z_{\mathcal{X},v_0} Z_{\mathcal{X},v_1} \right] - E \left[ Z_{\mathcal{X},v_0} \right] E \left[ Z_{\mathcal{X},v_1} \right] \\
&= E \left[ Z_{\mathcal{X},v_0} Z_{\mathcal{X},v_1} \right] - \left[ P_{B_0} + P_{B_1} \right](v_0) \left[ P_{B_0} + P_{B_1} \right](v_1).
\end{aligned}
\end{equation}
Again, in the double sum over $(v_0,v_1)$ in~\eqref{eq:asyvar}, the two cases from~\eqref{eq:E_Z_X_v0_Z_X_v1}, implicitly inherited by~\eqref{eq:sigZ_X}, apply on and off the diagonal, respectively.

To evaluate the second identity of~\eqref{eq:mixedadd}, we derive its two mixed moments separately:
\begin{equation}\label{eq:CZ}
\begin{aligned}
E \left[ \tilde{C}_U \tilde{Z}_{U,v} \right] &= E \left[ \tilde{C}_U Z_{U,v} \right] - \cancel{E \left[ \tilde{C}_U \right] E \left[ Z_{U,v} \right]} \\
&= P_{B_U} \left( \lambda \left( \bar{P}(v) \right) - E \left[ \lambda \left( \bar{P}(B_U) \right) \right] \right), \\
E \left[ \tilde{C}_{\mathcal{X}} \tilde{Z}_{\mathcal{X},v} \right]
&= E \left[ \tilde{C}_{\mathcal{X}} Z_{\mathcal{X},v} \right] - \cancel{E \left[ \tilde{C}_{\mathcal{X}} \right] E \left[ Z_{\mathcal{X},v} \right]} \\
&= P_{B_0}(v) E \left[ \lambda \left( \bar{P}(B_1) \right) \middle| B_0 = v \right] + P_{B_1}(v) E \left[ \lambda \left( \bar{P}(B_0) \right) \middle| B_1 = v \right] \\
&\quad + \left[ P_{B_0} + P_{B_1} \right](v) \Bigl( \lambda \left( \bar{P}(v) \right) - E \left[ \lambda \left( \bar{P}(B_0) \right) + \lambda \left( \bar{P}(B_1) \right) \right] \Bigr).
\end{aligned}
\end{equation}

With one component of the exponent pair fixed, each conditional expectation in~\eqref{eq:CZ} averages the surprisal of the admissible values of the other component under $\bar{P}$.
The relevant row or column of~\eqref{eq:cells} then supplies the required weights:
\begin{equation}\label{eq:partnervals}
\begin{aligned}
E \left[ \lambda \left( \bar{P}(B_1) \right) \middle| B_0 = v \right] &= \frac{\displaystyle \sum_{\gamma_1 \in \Gamma} n_{\gamma_v \gamma_1} \lambda \left( \bar{P}(\gamma_1) \right)}{\displaystyle \sum_{\gamma_1 \in \Gamma} n_{\gamma_v \gamma_1}}, \\
E \left[ \lambda \left( \bar{P}(B_0) \right) \middle| B_1 = v \right] &= \frac{\displaystyle \sum_{\gamma_0 \in \Gamma} n_{\gamma_0 \gamma_v} \lambda \left( \bar{P}(\gamma_0) \right)}{\displaystyle \sum_{\gamma_0 \in \Gamma} n_{\gamma_0 \gamma_v}}.
\end{aligned}
\end{equation}
Here $\gamma_v\in\Gamma$ denotes the category containing $v$.

Finally, the third identity of~\eqref{eq:mixedadd} reduces to the following expansions:
\begin{equation}\label{eq:CZ2}
\begin{aligned}
E \left[ \tilde{C}_U \tilde{Z}_{U,v}^{2} \right] &= E \left[ \tilde{C}_U Z_{U,v}^2 \right] - 2 E \left[ \tilde{C}_U Z_{U,v} \right] E \left[ Z_{U,v} \right] + \cancel{E \left[ \tilde{C}_U \right] E \left[ Z_{U,v} \right]^2} \\
&= (1 - 2 P_{B_U}) E \left[ \tilde{C}_U \tilde{Z}_{U,v} \right], \\
E \left[ \tilde{C}_{\mathcal{X}} \tilde{Z}_{\mathcal{X},v}^{2} \right]
&= E \left[ \tilde{C}_{\mathcal{X}} Z_{\mathcal{X},v}^{2} \right] - 2 E \left[ \tilde{C}_{\mathcal{X}} Z_{\mathcal{X},v} \right] E \left[ Z_{\mathcal{X},v} \right] + \cancel{E \left[ \tilde{C}_{\mathcal{X}} \right] E \left[ Z_{\mathcal{X},v} \right]^{2}} \\
&= \Lambda_{\mathcal{X},v}(1) E \left[ \tilde{C}_{\mathcal{X}} \middle| Z_{\mathcal{X},v} = 1 \right] + 4 \Lambda_{\mathcal{X},v}(2) E \left[ \tilde{C}_{\mathcal{X}} \middle| Z_{\mathcal{X},v} = 2 \right] \\
&\quad - 2 \left[ P_{B_0} + P_{B_1} \right](v) E \left[ \tilde{C}_{\mathcal{X}} \tilde{Z}_{\mathcal{X},v} \right] \\
&= \bigl( 1 - 2 \left[ P_{B_0} + P_{B_1} \right](v) \bigr) E \left[ \tilde{C}_{\mathcal{X}} \tilde{Z}_{\mathcal{X},v} \right] + 2 \Lambda_{\mathcal{X},v}(2) \, \tilde{C}_{\mathcal{X}}(v, v).
\end{aligned}
\end{equation}

Substituting~\eqref{eq:varCnum}, \eqref{eq:mixedadd}, \eqref{eq:sigZ_U}, \eqref{eq:sigZ_X}, \eqref{eq:CZ}, and~\eqref{eq:CZ2} into~\eqref{eq:asyvar} gives
\begin{equation}\label{eq:varconstants}
\begin{gathered}
\frac{\operatorname{Var}[C]}{50} \approx 0.000235108088, \\
\frac{1}{50 \ln(2)} \sum_{v \in \mathcal{V}} \frac{E \left[ \tilde{C} \tilde{X}_v^{2} \right]}{\bar{P}(v)} \approx 0.063911197090, \\
\frac{1}{2500 \ln(2)} \sum_{v \in \mathcal{V}}\frac{E \left[ \tilde{C} \tilde{X}_v \right] E \left[ \tilde{X}_v^{2} \right]}{\bar{P}(v)^{2}} \approx 0.064022119277, \\
\frac{1}{5000 \ln(2)^2} \sum_{v_0, v_1 \in \mathcal{V}} \frac{E \left[ \tilde{X}_{v_0} \tilde{X}_{v_1} \right]^{2}}{\bar{P}(v_0) \bar{P}(v_1)} \approx 265.373964301715.
\end{gathered}
\end{equation}
After assembling these coefficients accordingly, we obtain the numerical expansion of the reference variance
\begin{equation}\label{eq:varassembled}
\boxed{\displaystyle \left[ \sigma^{\mathrm{VSTL}}(T) \right]^2 = \frac{0.000235108088}{N} + \frac{265.374075223902}{N^2} + O \left(N^{-3} \right)}.
\end{equation}

\subsubsection{Consistency with the Uniform Model}\label{subsec:asyuniform}

\begin{corollary}\label{cor:asyuniform}
Under the uniform byte model, with $L=\lvert\mathcal{V}\rvert=256$, Propositions~\ref{prop:asymean} and~\ref{prop:asyvar} reduce to
\begin{equation}\label{eq:uniformvalues}
\begin{gathered}
\mu^{\mathrm{uniform}}(T) = 8 - \frac{L - 1}{2 N \ln(2)} - \frac{L^{2} - 1}{12 N^{2} \ln(2)} + O \left( N^{-3} \right), \\
\left[ \sigma^{\mathrm{uniform}}(T) \right]^{2} = \frac{L - 1}{2 N^{2} \ln(2)^{2}} + O \left( N^{-3} \right).
\end{gathered}
\end{equation}
\end{corollary}

\begin{proof}
Assume uniformity --- then the validity constraint no longer couples the exponent bytes $B_0$ and $B_1$ in each float, so they become independent, while every byte marginal, along with $\bar{P}$, coincides with $P_{B_U}=\dfrac{1}{L}$.
Thus, a block consists of $N$ independent uniform bytes, and its entropy converges to $H_\infty^{\mathrm{uniform}}=\log_2 L=8$.
For finite $N$, however, concavity makes the plug-in entropy biased downward~\cite[Thm.~2.7.3]{cover_thomas_2006}, so its reference mean requires a correction, while its reference variance is determined by fluctuations in the byte counts.

Under this uniform model, for any $v\in\mathcal{V}$, $X_v\sim\mathrm{Binomial}(50,P_{B_U})$, so $E[X_v]=50P_{B_U}$.
Moreover, every byte contributes the constant surprisal $\lambda(P_{B_U})$, giving $C=50H_\infty^{\mathrm{uniform}}=E[C]$.
It follows that $\operatorname{Var}[C]=0$ and $\tilde{C}=0$; therefore, all moments containing $\tilde{C}$ vanish.

Substituting the binomial moments of $X_v$ into~\eqref{eq:asymean} yields the mean in~\eqref{eq:uniformvalues}.
This expansion agrees, at the corresponding orders, with the results of Miller~\cite{miller_1955}, Basharin~\cite[Eq.~(3)]{basharin_1959}, and, after conversion to bits, Harris~\cite[Eq.~(14)]{harris_1977}.

For distinct byte values $v_0\ne v_1$, independence of the $50$ uniform bytes gives the record-count covariance $E[\tilde{X}_{v_0}\tilde{X}_{v_1}]=-50P_{B_U}^{2}$.
Substituting this covariance and the binomial moments of $X_v$ into~\eqref{eq:asyvar} yields the variance in~\eqref{eq:uniformvalues}.
After conversion to bits, this expansion agrees with Harris~\cite[Eq.~(23)]{harris_1977}; Basharin~\cite[Eq.~(4)]{basharin_1959} likewise notes that the $N^{-1}$ term vanishes for equiprobable symbols.
\end{proof}

\subsection{Normal Approximation and Acceptance Bounds}\label{subsec:asynormal}

\begin{corollary}\label{cor:clt}
As $T \to \infty$,
\begin{equation}\label{eq:clt}
\frac{H(S) - \mu^{\mathrm{VSTL}}(T)}{\sigma^{\mathrm{VSTL}}(T)} \xrightarrow{D} \mathcal{N}(0, 1),
\end{equation}
so, with $\Phi$ the standard normal distribution function, the two-sided acceptance interval for the global measure at significance level $\alpha$ is
\begin{equation}\label{eq:interval}
\mu^{\mathrm{VSTL}}(T) \pm \Phi^{-1} \left( 1 - \frac{\alpha}{2} \right) \sigma^{\mathrm{VSTL}}(T),
\end{equation}
where $\mu^{\mathrm{VSTL}}(T)$ and $\sigma^{\mathrm{VSTL}}(T)$ are the reference values of Propositions~\ref{prop:asymean} and~\ref{prop:asyvar}.
\end{corollary}

\begin{proof}
Collect the centred record counts into the vector $\tilde{X}_{\langle r \rangle} = \bigl\langle \tilde{X}_{v,r} \bigr\rangle_{v \in \mathcal{V}}$.
These vectors are i.i.d.\ across records by Definition~\ref{def:vstl}, and equation~\eqref{eq:devsum} holds simultaneously in all $L = \lvert \mathcal{V} \rvert$ coordinates.
The central limit theorem for i.i.d.\ summands~\cite{lindeberg_1922},~\cite[Thm.~(5.10.4)]{grimmett_stirzaker_2001}, in its multivariate form in $\mathbb{R}^{L}$~\cite[Ex.~2.18]{vandervaart_1998}, therefore gives
\begin{equation}\label{eq:histclt}
\begin{gathered}
\sqrt{N} \delta_S = \frac{1}{\sqrt{N}} \sum_{r=1}^{T} \tilde{X}_{\langle r \rangle} \xrightarrow{D} \mathcal{N}_{L} \left( 0, \Sigma \right), \\
\Sigma = \left\langle \frac{E \left[ \tilde{X}_{v_0} \tilde{X}_{v_1} \right]}{50} \right\rangle_{v_0, v_1 \in \mathcal{V}},
\end{gathered}
\end{equation}
where $\Sigma$ is the finite limiting covariance matrix formed from the per-record moments~\eqref{eq:mixedadd} already evaluated for Proposition~\ref{prop:asyvar}.

For each $v \in \mathcal{V}$, $\bar{P}(v) > 0$ and lies in the domain on which $\varphi$ is smooth; therefore, entropy is differentiable at the expansion point.
Since entropy is separable across byte values, the $L$ entries of its gradient are the $L$ evaluations of the one-variable derivative~\eqref{eq:phider}, exactly as in~\eqref{eq:htaylor}.
Thus,
\begin{equation}\label{eq:gradient}
\nabla H \left( \bar{P} \right) = \left\langle \varphi' \left( \bar{P}(v) \right) \right\rangle_{v \in \mathcal{V}} = \left\langle \lambda \left( \bar{P}(v) \right) - \frac{1}{\ln(2)} \right\rangle_{v \in \mathcal{V}} .
\end{equation}

Because $H$ is differentiable at $\bar{P}$, the delta method~\cite[Thm.~3.1]{vandervaart_1998} transforms the multivariate convergence in~\eqref{eq:histclt} into a univariate normal limit. Its variance is the quadratic form $\nabla H(\bar{P})^{T} \Sigma \nabla H(\bar{P})$, which projects the histogram covariance along the entropy gradient.
Since $H(S) = H(Q_S)$ and $H(\bar{P}) = H_\infty^{\mathrm{VSTL}}$ by Proposition~\ref{prop:ceiling}, this gives
\begin{equation}\label{eq:deltasandwich}
\sqrt{N} \left( H(S) - H_\infty^{\mathrm{VSTL}} \right) \xrightarrow{D} \mathcal{N} \left( 0, \nabla H \left( \bar{P} \right)^{T} \Sigma \nabla H \left( \bar{P} \right) \right).
\end{equation}

To evaluate the limiting variance in~\eqref{eq:deltasandwich}, consider the inner product of the gradient~\eqref{eq:gradient} with the centred byte counts $\tilde{X}_{v}$ for a generic record.
This calculation follows an approach similar to~\eqref{eq:linterm}:
\begin{equation}\label{eq:gradcontract}
\begin{aligned}
\sum_{v \in \mathcal{V}} \varphi' \left( \bar{P}(v) \right) \tilde{X}_{v}
&= \sum_{v \in \mathcal{V}} \left( \lambda \left( \bar{P}(v) \right) - \frac{1}{\ln(2)} \right) \tilde{X}_{v} \\
&= \sum_{v \in \mathcal{V}} \lambda \left( \bar{P}(v) \right) \tilde{X}_{v} - \frac{1}{\ln(2)} \sum_{v \in \mathcal{V}} \tilde{X}_{v} \\
&= \sum_{v \in \mathcal{V}} \lambda \left( \bar{P}(v) \right) X_{v} - E \left[ \sum_{v \in \mathcal{V}} \lambda \left( \bar{P}(v) \right) X_{v} \right] \\
&= C - E[C] = \tilde{C}.
\end{aligned}
\end{equation}
 The constant term in $\varphi'(x)$ is eliminated because $\displaystyle \sum_{v \in \mathcal{V}} X_v = \sum_{v \in \mathcal{V}} E[X_v] = 50$ and therefore $\displaystyle \sum_{v \in \mathcal{V}} \tilde{X}_v = 0$.

Substituting the entries of $\Sigma$ from~\eqref{eq:histclt} into the quadratic form and applying bilinearity allows it to be written in terms of the inner product in~\eqref{eq:gradcontract}:
\begin{equation}\label{eq:sandwich}
\begin{aligned}
\nabla H \left( \bar{P} \right)^{T} \Sigma \, \nabla H \left( \bar{P} \right)
&= \frac{1}{50} \sum_{v_0 \in \mathcal{V}} \sum_{v_1 \in \mathcal{V}} \varphi' \left( \bar{P}(v_0) \right) E \left[ \tilde{X}_{v_0} \tilde{X}_{v_1} \right] \varphi' \left( \bar{P}(v_1) \right) \\
&= \frac{1}{50} E \left[ \left( \sum_{v_0 \in \mathcal{V}} \varphi' \left( \bar{P}(v_0) \right) \tilde{X}_{v_0} \right) \left( \sum_{v_1 \in \mathcal{V}} \varphi' \left( \bar{P}(v_1) \right) \tilde{X}_{v_1} \right) \right] \\
&= \frac{E \left[ \tilde{C}^{2} \right]}{50} = \frac{\operatorname{Var}[C]}{50}.
\end{aligned}
\end{equation}
Equation~\eqref{eq:deltasandwich} therefore becomes
\begin{equation}\label{eq:cltraw}
\sqrt{N} \left( H(S) - H_\infty^{\mathrm{VSTL}} \right) \xrightarrow{D} \mathcal{N} \left( 0, \frac{\operatorname{Var}[C]}{50} \right).
\end{equation}

Rewriting the left-hand side of~\eqref{eq:clt} in terms of the $\sqrt{N}$-scaled convergence about $H_\infty^{\mathrm{VSTL}}$ in~\eqref{eq:cltraw} gives
\begin{equation}\label{eq:cltquotient}
\frac{H(S) - \mu^{\mathrm{VSTL}}(T)}{\sigma^{\mathrm{VSTL}}(T)}
= \frac{\sqrt{N} \left( H(S) - H_\infty^{\mathrm{VSTL}} \right) - \sqrt{N} \left( \mu^{\mathrm{VSTL}}(T) - H_\infty^{\mathrm{VSTL}} \right)}{\sqrt{N} \, \sigma^{\mathrm{VSTL}}(T)} .
\end{equation}
The first term in the numerator converges in distribution by~\eqref{eq:cltraw}; the mean expansion~\eqref{eq:asymean} with the coefficients~\eqref{eq:meanconstants} shows that the second term in the numerator is $O \left( N^{-1/2} \right)$ and therefore vanishes as $T \to \infty$; the leading term of~\eqref{eq:asyvar} shows that the denominator converges to a positive constant:
\begin{equation}\label{eq:cltthree}
\begin{gathered}
\sqrt{N} \left( H(S) - H_\infty^{\mathrm{VSTL}} \right) \xrightarrow{D} \mathcal{N} \left( 0, \frac{\operatorname{Var}[C]}{50} \right), \\[16pt]
\sqrt{N} \left( \mu^{\mathrm{VSTL}}(T) - H_\infty^{\mathrm{VSTL}} \right) = O \left( N^{-1/2} \right) \to 0, \\[16pt]
N \left[ \sigma^{\mathrm{VSTL}}(T) \right]^2 = \frac{\operatorname{Var}[C]}{50} + O \left( N^{-1} \right) \to \frac{\operatorname{Var}[C]}{50}, \\
\therefore \sqrt{N} \sigma^{\mathrm{VSTL}}(T) \to \sqrt{\frac{\operatorname{Var}[C]}{50}}.
\end{gathered}
\end{equation}
Slutsky's theorem~\cite{slutsky_1925},~\cite[Exercise~(7.2.5)]{grimmett_stirzaker_2001}, applied to the numerator and denominator of~\eqref{eq:cltquotient}, together with the scaling rule for normal distributions~\cite[Example~(4.4.4)]{grimmett_stirzaker_2001}, therefore yields
\begin{equation}\label{eq:cltunit}
\frac{H(S) - \mu^{\mathrm{VSTL}}(T)}{\sigma^{\mathrm{VSTL}}(T)} \xrightarrow{D} \frac{\displaystyle \mathcal{N} \left( 0, \frac{\operatorname{Var}[C]}{50} \right)}{\displaystyle \sqrt{\frac{\operatorname{Var}[C]}{50}}} \sim \mathcal{N} \left( 0, 1 \right).
\end{equation}
This proves~\eqref{eq:clt}.

The central $1-\alpha$ region of $\mathcal{N}(0,1)$ is bounded by $\pm \Phi^{-1}(1-\alpha/2)$; applying these bounds to~\eqref{eq:clt} gives~\eqref{eq:interval}.
\end{proof}

\subsection{The Test Procedure}\label{subsec:procedure}

The second-order expansion and the normal approximation become increasingly accurate as $T$ grows.
At small triangle counts, however, the remainder of the expansions is too large to neglect.
Figure~\ref{fig:exactasyreference} shows the resulting errors decreasing over $1 \le T \le 200$, motivating a finite lower limit for the test that neither asymptotic result specifies.

\begin{figure}
\centering
\begin{subfigure}{\linewidth}
\centering
\makebox[\linewidth][l]{\small{\textbf{(a)}}}\\[0pt]
\includegraphics[width=\linewidth]{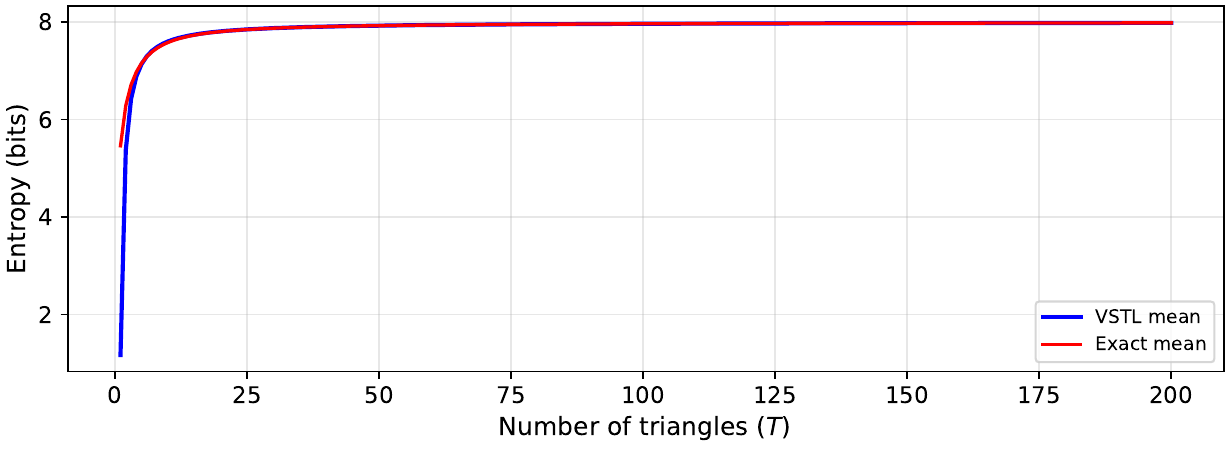}
\phantomsubcaption\label{fig:exactasymean}
\end{subfigure}
\begin{subfigure}{\linewidth}
\centering
\makebox[\linewidth][l]{\small{\textbf{(b)}}}\\[0pt]
\includegraphics[width=\linewidth]{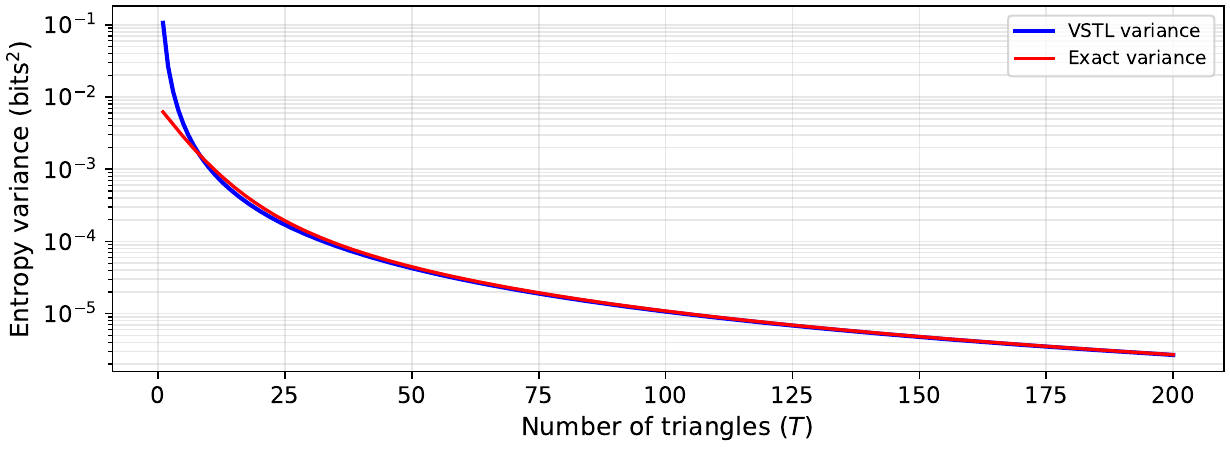}
\phantomsubcaption\label{fig:exactasyvar}
\end{subfigure}
\begin{subfigure}{\linewidth}
\centering
\makebox[\linewidth][l]{\small{\textbf{(c)}}}\\[0pt]
\includegraphics[width=\linewidth]{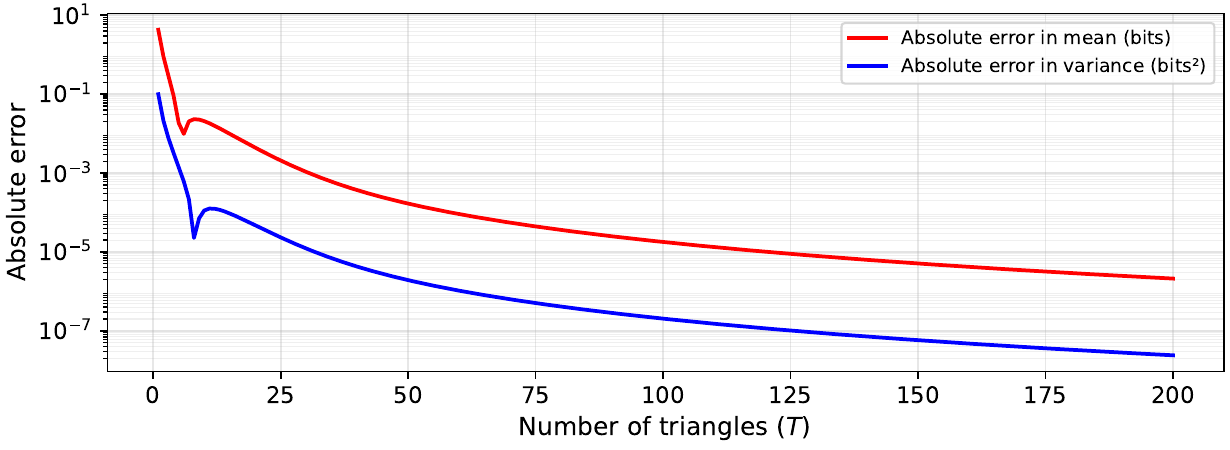}
\phantomsubcaption\label{fig:exactasyerror}
\end{subfigure}
\caption{Exact moments compared with the second-order VSTL expansions: (a) means; (b) variances; (c) absolute errors.}
\label{fig:exactasyreference}
\end{figure}

Figure~\ref{fig:exactasyintervals} compares intervals centred on the exact and second-order reference means and bounded at one respective standard deviation on either side.
The ratio of their intersection length to their union length equals approximately $99.1\%$ at $T=100$ and continues to increase, reaching $99.6\%$ at $T=200$.
We therefore adopt $T \ge 100$ as an operational lower limit for the test.

\begin{figure}
\centering
\begin{subfigure}{\linewidth}
\centering
\makebox[\linewidth][l]{\small{\textbf{(a)}}}\\[0pt]
\includegraphics[width=\linewidth]{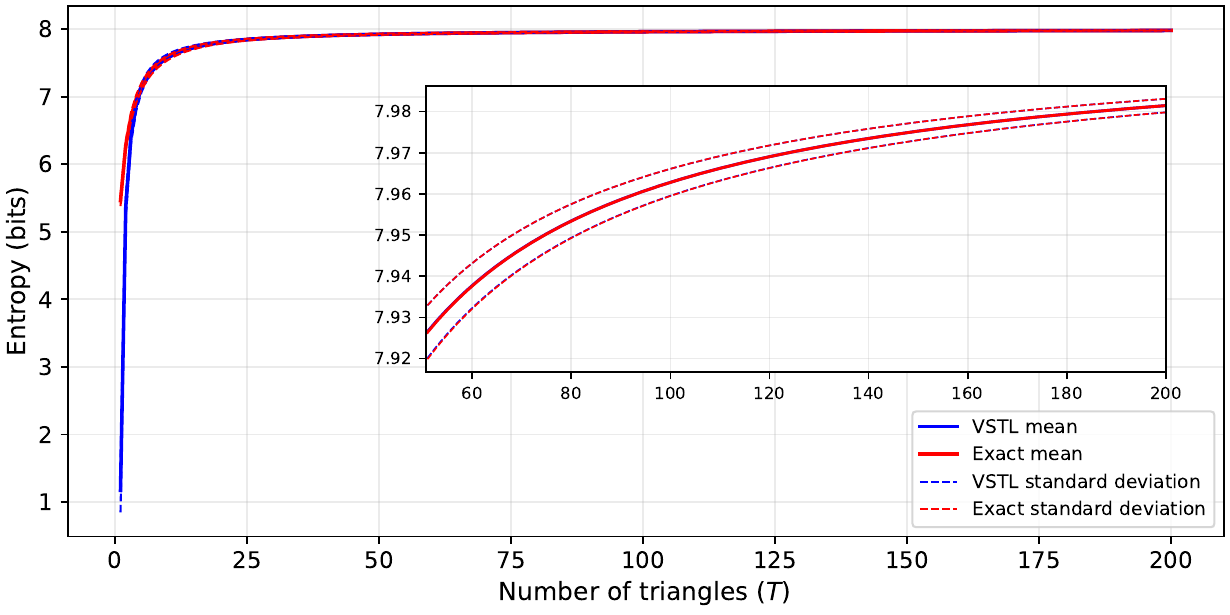}
\phantomsubcaption\label{fig:exactasyinterval}
\end{subfigure}
\begin{subfigure}{\linewidth}
\centering
\makebox[\linewidth][l]{\small{\textbf{(b)}}}\\[0pt]
\includegraphics[width=\linewidth]{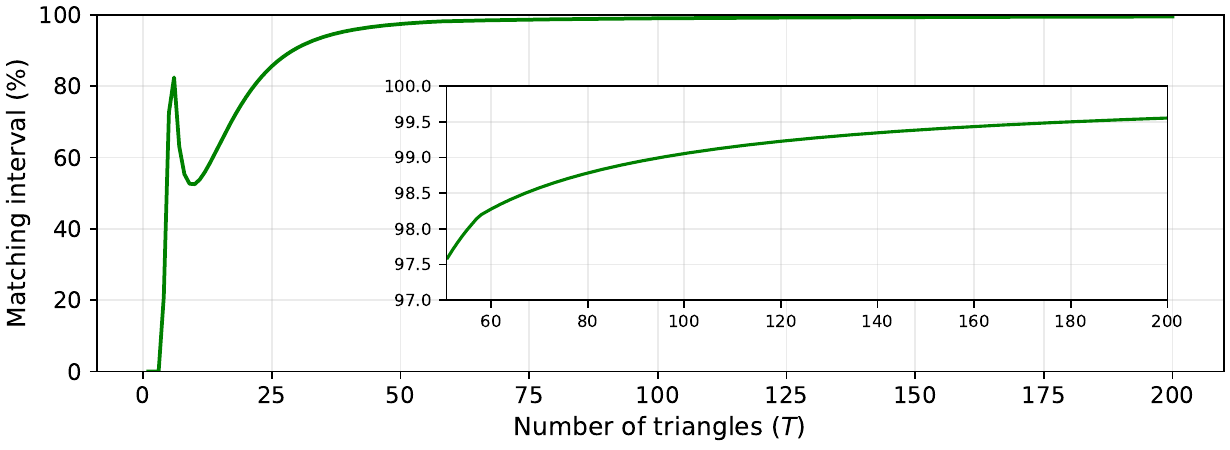}
\phantomsubcaption\label{fig:exactasycoverage}
\end{subfigure}
\caption{Comparison of the exact and second-order VSTL references: (a) intervals centred on the respective means, with dashed bounds at one standard deviation on either side; (b) percentage overlap between the intervals of (a).}
\label{fig:exactasyintervals}
\end{figure}

Corollary~\ref{cor:clt} turns the reference values of Propositions~\ref{prop:asymean} and~\ref{prop:asyvar} into a two-sided test for the entropy of a single ciphertext.
Table~\ref{tab:asym} lists the resulting mean, variance and acceptance bounds at three significance levels for selected numbers of triangles.

\begin{table}
\caption{Second-order reference moments and two-sided acceptance bounds for the VSTL global entropy test.}
\label{tab:asym}
\centering
\scriptsize
\setlength{\tabcolsep}{3pt}
\begin{tabular*}{\textwidth}{@{\extracolsep{\fill}}llcccc@{}}
\toprule
$T$ & $\alpha$ & $\mu^{\mathrm{VSTL}}(T)$ & $[\sigma^{\mathrm{VSTL}}(T)]^2$ & $H_\alpha^{-}(T)$ & $H_\alpha^{+}(T)$ \\
\midrule
\multirow{3}{*}{$100$} & $0.05$ & \multirow{3}{*}{$7.962810934102$} & \multirow{3}{*}{$0.000010661985$} & $7.956411123567$ & $7.969210744638$ \\
& $0.01$ & & & $7.954400157351$ & $7.971221710853$ \\
& $0.001$ & & & $7.952066477668$ & $7.973555390536$ \\
\midrule
\multirow{3}{*}{$1000$} & $0.05$ & \multirow{3}{*}{$7.996232777643$} & \multirow{3}{*}{$0.000000110852$} & $7.995580219145$ & $7.996885336141$ \\
& $0.01$ & & & $7.995375170404$ & $7.997090384883$ \\
& $0.001$ & & & $7.995137216092$ & $7.997328339194$ \\
\midrule
\multirow{3}{*}{$2000$} & $0.05$ & \multirow{3}{*}{$7.998074574860$} & \multirow{3}{*}{$0.000000028888$} & $7.997741447349$ & $7.998407702371$ \\
& $0.01$ & & & $7.997636771098$ & $7.998512378621$ \\
& $0.001$ & & & $7.997515296736$ & $7.998633852984$ \\
\midrule
\multirow{3}{*}{$5000$} & $0.05$ & \multirow{3}{*}{$7.999178896716$} & \multirow{3}{*}{$0.000000005186$} & $7.999037746408$ & $7.999320047024$ \\
& $0.01$ & & & $7.998993393767$ & $7.999364399666$ \\
& $0.001$ & & & $7.998941923555$ & $7.999415869878$ \\
\midrule
\multirow{3}{*}{$10000$} & $0.05$ & \multirow{3}{*}{$7.999546877923$} & \multirow{3}{*}{$0.000000001532$} & $7.999470170620$ & $7.999623585227$ \\
& $0.01$ & & & $7.999446067438$ & $7.999647688408$ \\
& $0.001$ & & & $7.999418096254$ & $7.999675659592$ \\
\midrule
\multirow{3}{*}{$20000$} & $0.05$ & \multirow{3}{*}{$7.999730844887$} & \multirow{3}{*}{$0.000000000500$} & $7.999686997633$ & $7.999774692140$ \\
& $0.01$ & & & $7.999673219828$ & $7.999788469945$ \\
& $0.001$ & & & $7.999657231004$ & $7.999804458769$ \\
\midrule
\multirow{3}{*}{$50000$} & $0.05$ & \multirow{3}{*}{$7.999841217500$} & \multirow{3}{*}{$0.000000000137$} & $7.999818318352$ & $7.999864116648$ \\
& $0.01$ & & & $7.999811122918$ & $7.999871312082$ \\
& $0.001$ & & & $7.999802772784$ & $7.999879662216$ \\
\midrule
\multirow{3}{*}{$100000$} & $0.05$ & \multirow{3}{*}{$7.999878007110$} & \multirow{3}{*}{$0.000000000058$} & $7.999863127307$ & $7.999892886914$ \\
& $0.01$ & & & $7.999858451734$ & $7.999897562487$ \\
& $0.001$ & & & $7.999853025840$ & $7.999902988381$ \\
\midrule
\multirow{3}{*}{$200000$} & $0.05$ & \multirow{3}{*}{$7.999896401679$} & \multirow{3}{*}{$0.000000000026$} & $7.999886376210$ & $7.999906427149$ \\
& $0.01$ & & & $7.999883225979$ & $7.999909577379$ \\
& $0.001$ & & & $7.999879570209$ & $7.999913233149$ \\
\midrule
\multirow{3}{*}{$500000$} & $0.05$ & \multirow{3}{*}{$7.999907438345$} & \multirow{3}{*}{$0.000000000010$} & $7.999901293640$ & $7.999913583050$ \\
& $0.01$ & & & $7.999899362833$ & $7.999915513856$ \\
& $0.001$ & & & $7.999897122178$ & $7.999917754512$ \\
\midrule
\multirow{3}{*}{$1000000$} & $0.05$ & \multirow{3}{*}{$7.999911117221$} & \multirow{3}{*}{$0.000000000005$} & $7.999906819439$ & $7.999915415003$ \\
& $0.01$ & & & $7.999905468978$ & $7.999916765464$ \\
& $0.001$ & & & $7.999903901799$ & $7.999918332642$ \\
\bottomrule
\end{tabular*}
\end{table}

Algorithm~\ref{alg:entropy} applies the reference expansions of Propositions~\ref{prop:asymean} and~\ref{prop:asyvar} and the acceptance interval of Corollary~\ref{cor:clt} to a single ciphertext.

\begin{algorithm}
\caption{VSTL global entropy test}\label{alg:entropy}
\begin{algorithmic}[1]
\Require ciphertext $S$ consisting of T triangles, with $N = 50T$ data bytes $S_1, \ldots, S_N$; significance level $\alpha$
\Ensure global entropy acceptance verdict under the VSTL reference
\State $\#(v) \gets 0$ for all $v \in \mathcal{V}$
\For{$k \gets 1$ \textbf{to} $N$}
\State $\#(S_k) \gets \#(S_k) + 1$
\EndFor
\State $H \gets 0$
\For{$v \in \mathcal{V}$ such that $\#(v) > 0$}
\State $H \gets H - \dfrac{\#(v)}{N} \log_2 \dfrac{\#(v)}{N}$ \Comment{\eqref{eq:plugin}}
\EndFor
\State $\mu \gets 7.999914796090 - \dfrac{183.943323730404}{N} - \dfrac{7879.931052184461}{N^2}$ \Comment{\eqref{eq:meanassembled}}
\State $\sigma \gets \sqrt{\dfrac{0.000235108088}{N} + \dfrac{265.374075223902}{N^2}}$ \Comment{\eqref{eq:varassembled}}
\State $H^{-} \gets \mu - \Phi^{-1}(1 - \alpha/2) \sigma$; \quad $H^{+} \gets \mu + \Phi^{-1}(1 - \alpha/2) \sigma$ \Comment{\eqref{eq:interval}}
\State \Return $H^{-} \le H \le H^{+}$
\end{algorithmic}
\end{algorithm}

\subsection{Practical Significance}\label{subsec:globalpractical}

As $T$ grows, the difference between the uniform and VSTL reference means approaches the deficit $\Delta$ of~\eqref{eq:deficit}, while their variability decreases.
We measure this separation in units of the uniform deviation to determine when the uniform reference becomes unreliable in practice.
The crossing condition is
\begin{equation}\label{eq:uniformcrossing}
\mu^{\mathrm{uniform}}(T)-\mu^{\mathrm{VSTL}}(T)=z\sigma^{\mathrm{uniform}}(T).
\end{equation}
Solving~\eqref{eq:uniformcrossing} for $T$ shows that the mean gap first exceeds one uniform standard deviation at $T=3824$ and the two-sided $\alpha=0.05$ acceptance half-width at $T=7495$.
Both thresholds lie within practically relevant file size ranges.

\subsection{Monte-Carlo Validation}\label{subsec:globalmc}

Following the simulation procedure of~\cite{adjusted_npcr_uaci_2026}, we check whether samples from the VSTL distribution recover the global entropy reference moments.

For a fixed triangle count $T$, one trial draws a VSTL file according to Definition~\ref{def:vstl}.
Each of its $12T$ floats is sampled independently and uniformly from $\mathcal{F}$, and each attribute word independently and uniformly from $\{0,\ldots,2^{16}-1\}$.
We compute the global entropy using~\eqref{eq:plugin} over the bytes of the triangle records.
For each of the selected file sizes spanning $500 \le T \le 1000000$, $n = 10000$ independent trials are performed.

The resulting observations $H_1,\ldots,H_n$ give the sample mean and unbiased sample variance,
\begin{equation}\label{eq:mcestimators}
\begin{gathered}
\hat{\mu} = \frac{1}{n} \sum_{i=1}^{n} H_i, \\
\hat{\sigma}^{2} = \frac{1}{n - 1} \sum_{i=1}^{n} \left( H_i - \hat{\mu} \right)^{2}.
\end{gathered}
\end{equation}
These estimators converge to the true moments as $n$ grows, providing a numerical check on the second-order reference values of Propositions~\ref{prop:asymean} and~\ref{prop:asyvar}.

To compare the reference values and Monte-Carlo estimates on a common scale, we use the standardized displacement
\begin{equation}\label{eq:zeta}
\zeta(x) = \frac{x - \mu^{\mathrm{VSTL}}(T)}{\sigma^{\mathrm{VSTL}}(T)},
\end{equation}
which expresses the difference between any $x$ and the reference mean in reference standard deviations.

Figure~\ref{fig:globalmc} confirms agreement between the simulated and predicted VSTL reference values --- the mean and standard deviation follow their theoretical curves within sampling error at every tested size.
It also reveals the consequences of wrongly relying on the uniform byte null model for VSTL ciphertexts.
The vertical lines at $T=3824$ and $T=7495$ mark the crossings derived in Section~\ref{subsec:globalpractical}, where the mean difference equals one uniform standard deviation and the uniform $95\%$ acceptance half-width, respectively.
Beyond the second threshold, the lower acceptance bound of the conventional test, based on the uniform byte null model, exceeds the VSTL reference mean, making the test unsuitable even at practically relevant file sizes and reinforcing the need for the format-specific test proposed in this paper.

Figure~\ref{fig:globaldiff} shows the absolute differences between the reference values and their estimates on a logarithmic scale; both decrease towards zero over the tested range, with no apparent persistent discrepancy.

\begin{figure}
\centering
\begin{subfigure}{\linewidth}
\centering
\makebox[\linewidth][l]{\small{\textbf{(a)}}}\\[0pt]
\includegraphics[width=\linewidth]{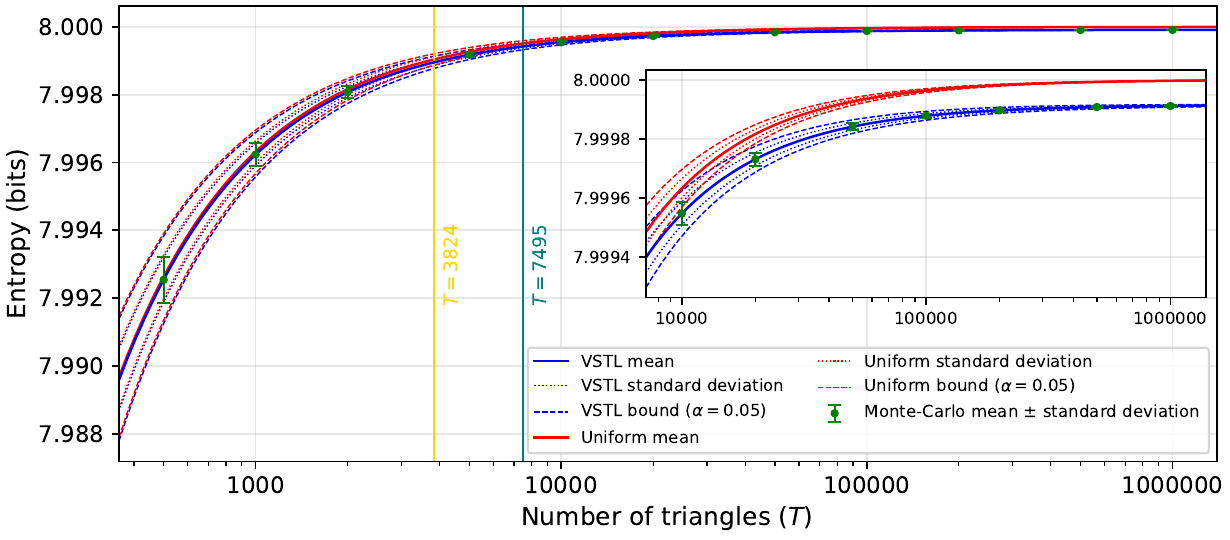}
\phantomsubcaption\label{fig:globalmc-raw}
\end{subfigure}
\begin{subfigure}{\linewidth}
\centering
\makebox[\linewidth][l]{\small{\textbf{(b)}}}\\[0pt]
\includegraphics[width=\linewidth]{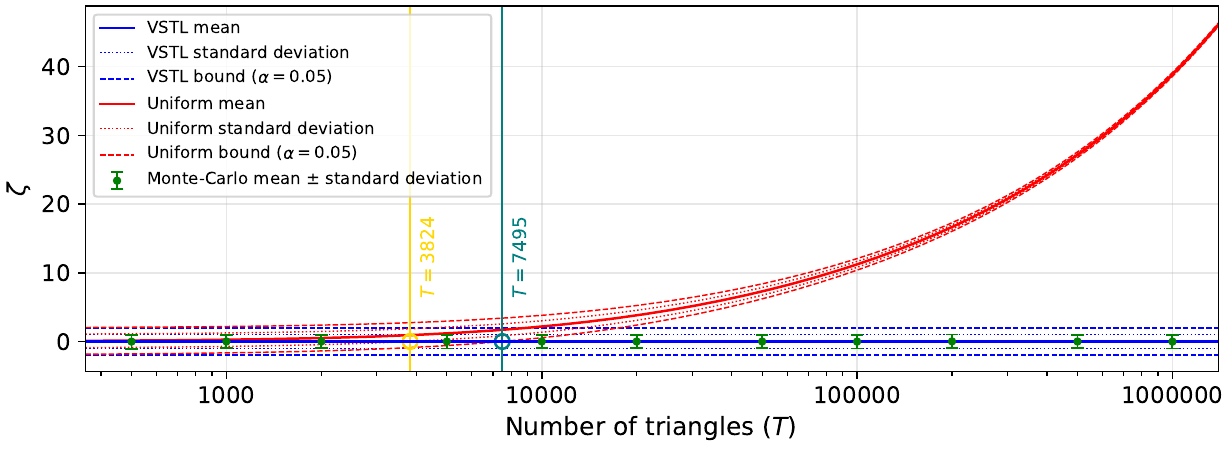}
\phantomsubcaption\label{fig:globalmc-std}
\end{subfigure}
\caption{Global entropy under the VSTL and uniform null models: (a) reference values from the second-order expansions and Monte-Carlo estimates; (b) the standardized counterpart of (a), expressed in $\zeta$-units.}
\label{fig:globalmc}
\end{figure}

\begin{figure}
\centering
\includegraphics[width=\linewidth]{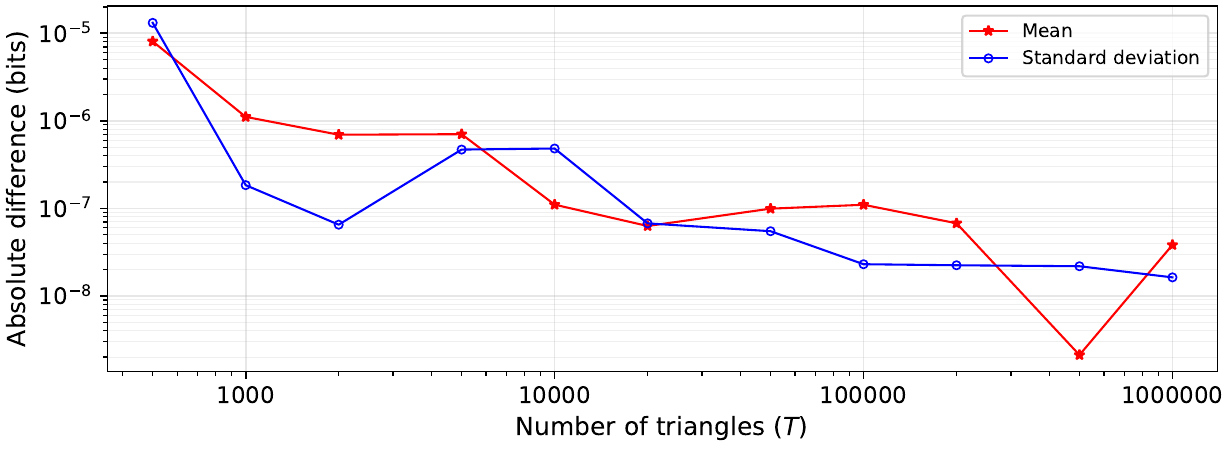}
\caption{Absolute difference between the reference values of the VSTL global entropy test and their corresponding Monte-Carlo estimates: $\lvert \mu^{\mathrm{VSTL}}(T) - \hat{\mu} \rvert$ (red stars) and $\lvert \sigma^{\mathrm{VSTL}}(T) - \hat{\sigma} \rvert$ (blue circles).}
\label{fig:globaldiff}
\end{figure}

\section{The VSTL Local Entropy Test}\label{sec:local}

The $(k, N_B)$-local test selects $k$ non-overlapping blocks of $T_B$ consecutive records from a ciphertext $S$, denoted by $S_1, \ldots, S_k \in \mathcal{V}^{N_B}$, and averages their Shannon entropies:
\begin{equation}\label{eq:local}
\bar{H}_{k, N_B}(S) = \frac{1}{k} \sum_{i=1}^{k} H(S_i).
\end{equation}
The conventional local entropy test evaluates the observed average against reference values derived under the uniform null model, which treats each block as a string of $N_B$ i.i.d.\ uniform bytes.
The VSTL local entropy test, however, requires reference values that account for the validity constraint.

This section adapts the capacity energy criterion of~\cite{local_entropy_2013} to determine the optimal block size $N_B^{\ast}$ for the VSTL local entropy test and then derives the exact reference mean and variance of $\bar{H}_{k, N_B^{\ast}}$ at that block size.

\begin{lemma}\label{lem:aligned}
Let $S_1, \ldots, S_k$ be pairwise disjoint blocks of $T_B$ consecutive records from a VSTL ciphertext. The average of their entropies has mean and variance
\begin{equation}\label{eq:localmoments}
\begin{gathered}
E \left[ \bar{H}_{k, N_B}(S) \right] = \mu^{\mathrm{VSTL}}(T_B), \\
\operatorname{Var} \left[ \bar{H}_{k, N_B}(S) \right] = \frac{\left[ \sigma^{\mathrm{VSTL}}(T_B) \right]^{2}}{k}.
\end{gathered}
\end{equation}
\end{lemma}

\begin{proof}
Under Definition~\ref{def:vstl}, records are independent, so blocks containing disjoint sets of records have independent entropies.
Every block contains the $T_B$ records drawn from the same distribution, so its entropy has mean $\mu^{\mathrm{VSTL}}(T_B)$ and variance $\left[ \sigma^{\mathrm{VSTL}}(T_B) \right]^2$.
Averaging the $k$ entropies preserves this mean and divides the variance by $k$, giving~\eqref{eq:localmoments}.
\end{proof}

\subsection{The Optimal Block Size}\label{subsec:blocksize}

The block size is determined by maximizing the capacity energy criterion of~\cite[Eq.~(28)]{local_entropy_2013}.
This criterion is the probability that every byte value occurs in a random block, divided by the block size $N_B$:
\begin{equation}\label{eq:capacity}
f(N_B) = \frac{\displaystyle \Pr \left[ \bigcap_{v \in \mathcal{V}} \left\{ \#(v) > 0 \right\} \right]}{N_B}.
\end{equation}

For the VSTL local entropy test, each block consists of $T_B$ triangle records, so $N_B = 50T_B$.
The probability that every byte value occurs at least once in a random block of size $N_B$ is calculated by inclusion--exclusion over the events of absence.
All values in a subset $\mathcal{A} \subseteq \mathcal{V}$ are absent from the block if and only if every byte lies outside $\mathcal{A}$:
\begin{equation}\label{eq:cover}
\begin{gathered}
\Pr \left[ \bigcap_{v \in \mathcal{V}} \left\{ \#(v) > 0 \right\} \right] = \sum_{\mathcal{A} \subseteq \mathcal{V}} (-1)^{\lvert \mathcal{A} \rvert} \Pr \left[ \bigcap_{v \in \mathcal{A}} \left\{ \#(v) = 0 \right\} \right] \\
= \sum_{\mathcal{A} \subseteq \mathcal{V}} (-1)^{\lvert \mathcal{A} \rvert} \Pr \left[ B_U \in \left( \mathcal{V} \setminus \mathcal{A} \right) \right]^{26T_B} \Pr \left[ (B_0, B_1) \in \left( \mathcal{V} \setminus \mathcal{A} \right)^{2} \right]^{12T_B}.
\end{gathered}
\end{equation}

Since $B_U$ is uniform on $\mathcal{V}$ and $(B_0, B_1)$ is uniform on $\mathcal{X}$, the two factors of~\eqref{eq:cover} can be expressed directly in terms of set cardinalities:
\begin{equation}\label{eq:piA}
\begin{gathered}
\Pr \left[ B_U \in \left( \mathcal{V} \setminus \mathcal{A} \right) \right] = \frac{\lvert \mathcal{V} \setminus \mathcal{A} \rvert}{\lvert \mathcal{V} \rvert} = \frac{256 - \lvert \mathcal{A} \rvert}{256} , \\
\Pr \left[ (B_0, B_1) \in \left( \mathcal{V} \setminus \mathcal{A} \right)^{2} \right] = \frac{\left\lvert \mathcal{X} \cap \left( \mathcal{V} \setminus \mathcal{A} \right)^{2} \right\rvert}{\lvert \mathcal{X} \rvert} = \frac{\left\lvert \mathcal{X} \cap \left( \mathcal{V} \setminus \mathcal{A} \right)^{2} \right\rvert}{65280} .
\end{gathered}
\end{equation}
To evaluate the numerator in the second line, we rewrite the intersection as a difference of products whose cardinalities can be calculated directly:
\begin{equation}\label{eq:admissible_pairs_outside_A}
\begin{gathered}
\begin{aligned}
\mathcal{X} \cap (\mathcal{V} \setminus \mathcal{A})^{2} &= \left[ \mathcal{V}^{2} \setminus \left( \{127, 255\} \times \mathcal{V}_{hi} \right) \right] \cap \left( \mathcal{V} \setminus \mathcal{A} \right)^{2} \\
&= \left( \mathcal{V} \setminus \mathcal{A} \right)^{2} \setminus \left[ \left( \{127, 255\} \setminus \mathcal{A} \right) \times \left( \mathcal{V}_{hi} \setminus \mathcal{A} \right) \right],
\end{aligned} \\[8pt]
\therefore \Pr \left[ (B_0, B_1) \in \left( \mathcal{V} \setminus \mathcal{A} \right)^{2} \right] = \frac{\left( 256 - \lvert \mathcal{A} \rvert \right)^{2} - \left\lvert \{127, 255\} \setminus \mathcal{A} \right\rvert \cdot \left\lvert \mathcal{V}_{hi} \setminus \mathcal{A} \right\rvert}{65280}.
\end{gathered}
\end{equation}

Equations~\eqref{eq:cover}, \eqref{eq:piA} and~\eqref{eq:admissible_pairs_outside_A} show that the contribution of $\mathcal{A}$ depends only on the number of its elements belonging to each of the four categories in~\eqref{eq:cats}.
The capacity energy criterion of~\eqref{eq:capacity} can thus be written as:
\begin{equation}\label{eq:coverfour}
\begin{gathered}
f(N_B) = \frac{1}{N_B} \sum_{\substack{0 \le n_{lo}, n_{hi} \le 127 \\ \varepsilon_{127}, \varepsilon_{255} \in \{ 0, 1 \}}} (-1)^{n_{lo} + n_{hi} + \varepsilon_{127} + \varepsilon_{255}} \binom{127}{n_{lo}} \binom{127}{n_{hi}} f_U^{26T_B} f_{\mathcal{X}}^{12T_B}, \\
f_U = \frac{256 - n_{lo} - n_{hi} - \varepsilon_{127} - \varepsilon_{255}}{256}, \\
f_{\mathcal{X}} = \frac{\left( 256 - n_{lo} - n_{hi} - \varepsilon_{127} - \varepsilon_{255} \right)^2 - \left( 2 - \varepsilon_{127} - \varepsilon_{255} \right) \left( 128 - n_{hi} - \varepsilon_{255} \right)}{65280},
\end{gathered}
\end{equation}
where $f_U$ and $f_{\mathcal{X}}$ are the probabilities that an unconstrained byte lies outside $\mathcal{A}$ and that both bytes of an exponent pair lie outside $\mathcal{A}$, respectively.
The variables $n_{lo}$, $\varepsilon_{127}$, $n_{hi}$ and $\varepsilon_{255}$ specify how many values from each category of $\Gamma$ belong to $\mathcal{A}$.
The binomial coefficients count the selections from the majority sets.

Following~\cite[Eq.~(31)]{local_entropy_2013}, the optimal block size is obtained by maximizing the criterion of~\eqref{eq:capacity}, evaluated using~\eqref{eq:coverfour}:
\begin{equation}\label{eq:tstar}
\begin{gathered}
T_B^{\ast} = \argmax_{T_B \ge 1} f(50 T_B) = 39, \\
N_B^{*} = 50 T_B^{\ast} = 1950.
\end{gathered}
\end{equation}
The capacity energy criterion is presented numerically near the maximum in Table~\ref{tab:capacity} and graphically over a wider range in Figure~\ref{fig:capacity}.

\begin{table}
\caption{The VSTL capacity energy criterion around the optimum. The underlined value marks the maximum at $T_B^{\ast}=39$.}
\label{tab:capacity}
\centering
\footnotesize
\setlength{\tabcolsep}{5pt}
\begin{tabular}{rrr}
\toprule
$T_B$ & $N_B$ & $f(N_B)$ \\
\midrule
$36$ & $1800$ & $0.000443710557$ \\
$37$ & $1850$ & $0.000449342144$ \\
$38$ & $1900$ & $0.000452139725$ \\
$39$ & $1950$ & $\underline{0.000452610292}$ \\
$40$ & $2000$ & $0.000451203039$ \\
$41$ & $2050$ & $0.000448305921$ \\
$42$ & $2100$ & $0.000444247471$ \\
\bottomrule
\end{tabular}
\end{table}

\begin{figure}
\centering
\includegraphics[width=\linewidth]{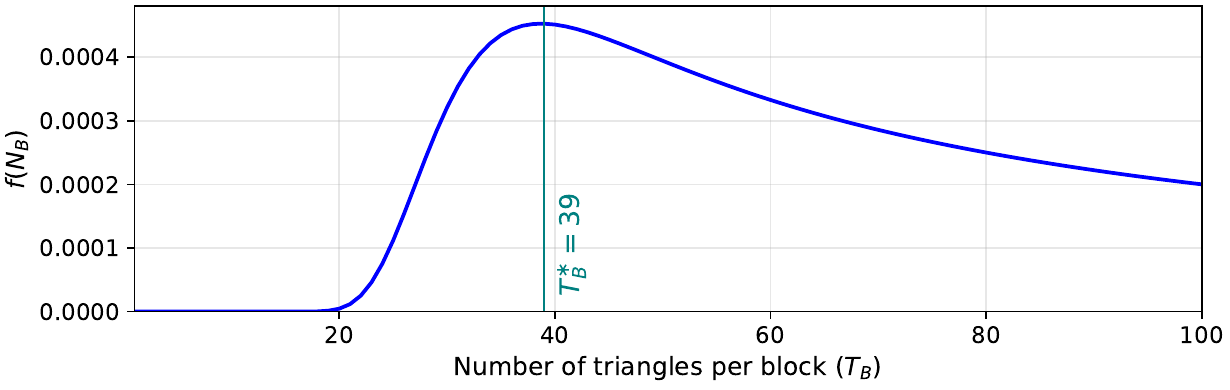}
\caption{The VSTL capacity energy criterion over $1 \le T_B \le 100$, evaluated at integer triangle counts. The vertical line marks the maximum at $T_B^{\ast} = 39$.}
\label{fig:capacity}
\end{figure}

\subsection{The Reference Values}\label{subsec:localref}

The second-order expansions of Propositions~\ref{prop:asymean} and~\ref{prop:asyvar} become increasingly accurate as $T \to \infty$.
However, the optimal local block size $T_B^{\ast} = 39$ lies well below the operational lower limit $T \ge 100$ adopted for the global test in Section~\ref{subsec:procedure}.
We therefore evaluate the exact count laws of Lemmas~\ref{lem:singular_count} and~\ref{lem:joint_count} once for $T_B^{\ast}$; the resulting mean and variance determine the reference values for every admissible $k$, avoiding the truncation errors of the global reference approximations.

Proposition~\ref{prop:moments} is evaluated at $T_B^{\ast} = 39$, with $26 T_B^{\ast} = 1014$ unconstrained bytes and $12 T_B^{\ast} = 468$ exponent pairs.
The count law of Lemma~\ref{lem:singular_count}, \eqref{eq:Eh} and~\eqref{eq:Eh2} then give the following moments:
\begin{equation}\label{eq:catmoments}
\begin{array}{c|cc}
& \theta_\gamma & \theta^{(2)}_\gamma \\
\hline
\mathcal{V}^{*}_{lo} & 0.030920079164 & 0.001042319790 \\
\{127\} & 0.027790753117 & 0.000852720838 \\
\mathcal{V}^{*}_{hi} & 0.030871842363 & 0.001039250977 \\
\{255\} & 0.027741945918 & 0.000849747236
\end{array}
\end{equation}
Likewise, the count law of Lemma~\ref{lem:joint_count} and~\eqref{eq:Ehh} give the joint moments for $v_0 \neq v_1$, with rows indexed by $\gamma_0$ and columns by $\gamma_1$:
\begingroup
\setlength{\arraycolsep}{4pt}
\begin{equation}\label{eq:pairmoments}
\begin{array}{c|cccc}
\theta_{\gamma_0 \gamma_1} & \mathcal{V}^{*}_{lo} & \{127\} & \mathcal{V}^{*}_{hi} & \{255\} \\
\hline
\mathcal{V}^{*}_{lo} & 0.000955712851 & 0.000859027963 & 0.000954222497 & 0.000857520016 \\
\{127\} & 0.000859027963 & \text{---} & 0.000857604860 & 0.000770677417 \\
\mathcal{V}^{*}_{hi} & 0.000954222497 & 0.000857604860 & 0.000952734466 & 0.000856099227 \\
\{255\} & 0.000857520016 & 0.000770677417 & 0.000856099227 & \text{---}
\end{array}
\end{equation}
\endgroup
The table is symmetric, and its two singleton diagonal entries are empty because a singleton category admits no pair of distinct values.

Equations~\eqref{eq:mean} and~\eqref{eq:var} assemble these moments into the exact block mean and variance.
By Lemma~\ref{lem:aligned}, averaging over $k$ blocks gives the VSTL counterparts of the reference values of~\cite[Tab.~3]{local_entropy_2013}:
\begin{equation}\label{eq:localref}
\boxed{%
\begin{gathered}
\mu_{\bar{H}}^{\mathrm{VSTL}} \approx 7.903106732975 \\
\sigma_{\bar{H}}^{\mathrm{VSTL}} \approx \frac{0.008635732664}{\sqrt{k}}
\end{gathered}
}.
\end{equation}

\subsection{Consistency with the Uniform Model}\label{subsec:localuniform}

\begin{corollary}\label{cor:localuniform}
Under the uniform byte model, with $L=\lvert\mathcal{V}\rvert=256$, Lemma~\ref{lem:aligned} reduces to
\begin{equation}\label{eq:localuniform}
\begin{gathered}
\mu_{\bar{H}}^{\mathrm{uniform}} = L \theta_{\mathcal{V}}, \\
\left[ \sigma_{\bar{H}}^{\mathrm{uniform}} \right]^{2} = \frac{L \theta^{(2)}_{\mathcal{V}} + L(L-1)\theta_{\mathcal{V}\mathcal{V}} - \left[ \mu_{\bar{H}}^{\mathrm{uniform}} \right]^{2}}{k},
\end{gathered}
\end{equation}
where the $\theta$ moments are evaluated at $N=N_B$.
\end{corollary}

\begin{proof}
Assume uniformity --- then Corollary~\ref{cor:uniform} gives the exact mean and variance of each block.
The $k$ disjoint blocks remain independent, so averaging preserves their mean and divides their variance by $k$.
Substituting~\eqref{eq:wu} into~\eqref{eq:localmoments} yields~\eqref{eq:localuniform}.

Equation~\eqref{eq:localuniform} therefore reproduces the reference mean and variance for the local entropy statistic of~\cite[Secs.~3.3 and~4.3]{local_entropy_2013}.
\end{proof}

\subsection{Normal Approximation and Acceptance Bounds}\label{subsec:localasynormal}

Unlike the global normal approximation of Corollary~\ref{cor:clt}, which lets the number of records in a block grow, the local approximation considers a growing number of blocks at a fixed block size.

\begin{corollary}\label{cor:localclt}
Let $S_1, \ldots, S_k$ be pairwise disjoint blocks of $T_B^{\ast}$ consecutive records from a VSTL ciphertext.
Their average entropy has reference mean $\mu_{\bar{H}}^{\mathrm{VSTL}}$ and standard deviation $\sigma_{\bar{H}}^{\mathrm{VSTL}}$ given in~\eqref{eq:localref}.
Then, as the block count $k$ grows,
\begin{equation}\label{eq:localclt}
\frac{\bar{H}_{k, N_B^{\ast}}(S) - \mu_{\bar{H}}^{\mathrm{VSTL}}}{\sigma_{\bar{H}}^{\mathrm{VSTL}}} \xrightarrow{D} \mathcal{N}(0, 1) ,
\end{equation}
so the two-sided acceptance interval at level $\alpha$ is
\begin{equation}\label{eq:localinterval}
\mu_{\bar{H}}^{\mathrm{VSTL}} \pm \Phi^{-1}\!\left( 1 - \frac{\alpha}{2} \right) \sigma_{\bar{H}}^{\mathrm{VSTL}}.
\end{equation}
\end{corollary}

\begin{proof}
By Lemma~\ref{lem:aligned}, the block entropies are independent and share the same distribution.
Since the variance in~\eqref{eq:localref} is finite and positive, the central limit theorem~\cite{lindeberg_1922},~\cite[Thm.~(5.10.4)]{grimmett_stirzaker_2001} gives~\eqref{eq:localclt}.

The central $1-\alpha$ region of $\mathcal{N}(0,1)$ is bounded by $\pm \Phi^{-1}(1-\alpha/2)$; applying these bounds to~\eqref{eq:localclt} gives~\eqref{eq:localinterval}.
\end{proof}

\subsection{The Test Procedure}\label{subsec:localtest}

The block count may be chosen within $30 \le k \le \left\lfloor \dfrac{T}{T_B^{\ast}} \right\rfloor$. The lower bound follows the convention of~\cite{local_entropy_2013} for the normal approximation, while the upper bound is the number of complete non-overlapping blocks available in a file of $T$ triangles. By~\eqref{eq:localmoments}, increasing $k$ reduces the standard deviation of the average, without changing its reference mean.

Corollary~\ref{cor:localclt} turns the reference values of~\eqref{eq:localref} into a two-sided test for the average entropy of $k$ non-overlapping blocks from a single ciphertext.
Table~\ref{tab:local} lists the resulting mean, standard deviation and acceptance bounds at three significance levels for selected block counts.

\begin{table}
\caption{Exact reference moments and two-sided acceptance bounds for the VSTL local entropy test at $T_B^{\ast} = 39$, $N_B^{\ast} = 1950$.}
\label{tab:local}
\centering
\scriptsize
\setlength{\tabcolsep}{3pt}
\begin{tabular*}{\textwidth}{@{\extracolsep{\fill}}llcccc@{}}
\toprule
$k$ & $\alpha$ & $\mu_{\bar{H}}^{\mathrm{VSTL}}$ & $\sigma_{\bar{H}}^{\mathrm{VSTL}}$ & $\bar{H}_\alpha^{-}$ & $\bar{H}_\alpha^{+}$ \\
\midrule
\multirow{3}{*}{$30$} & $0.05$ & \multirow{3}{*}{$7.903106732975$} & \multirow{3}{*}{$0.001576661860$} & $7.900016532513$ & $7.906196933436$ \\
 & $0.01$ &  &  & $7.899045521154$ & $7.907167944796$ \\
 & $0.001$ &  &  & $7.897918684977$ & $7.908294780972$ \\
\midrule
\multirow{3}{*}{$50$} & $0.05$ & \multirow{3}{*}{$7.903106732975$} & \multirow{3}{*}{$0.001221277025$} & $7.900713073990$ & $7.905500391960$ \\
 & $0.01$ &  &  & $7.899960931825$ & $7.906252534125$ \\
 & $0.001$ &  &  & $7.899088088276$ & $7.907125377673$ \\
\midrule
\multirow{3}{*}{$100$} & $0.05$ & \multirow{3}{*}{$7.903106732975$} & \multirow{3}{*}{$0.000863573266$} & $7.901414160475$ & $7.904799305475$ \\
 & $0.01$ &  &  & $7.900882315650$ & $7.905331150300$ \\
 & $0.001$ &  &  & $7.900265122057$ & $7.905948343892$ \\
\midrule
\multirow{3}{*}{$250$} & $0.05$ & \multirow{3}{*}{$7.903106732975$} & \multirow{3}{*}{$0.000546171690$} & $7.902036256134$ & $7.904177209816$ \\
 & $0.01$ &  &  & $7.901699887932$ & $7.904513578018$ \\
 & $0.001$ &  &  & $7.901309540430$ & $7.904903925520$ \\
\midrule
\multirow{3}{*}{$500$} & $0.05$ & \multirow{3}{*}{$7.903106732975$} & \multirow{3}{*}{$0.000386201705$} & $7.902349791541$ & $7.903863674408$ \\
 & $0.01$ &  &  & $7.902111943305$ & $7.904101522645$ \\
 & $0.001$ &  &  & $7.901835925939$ & $7.904377540010$ \\
\midrule
\multirow{3}{*}{$700$} & $0.05$ & \multirow{3}{*}{$7.903106732975$} & \multirow{3}{*}{$0.000326400015$} & $7.902467000702$ & $7.903746465248$ \\
 & $0.01$ &  &  & $7.902265982253$ & $7.903947483697$ \\
 & $0.001$ &  &  & $7.902032705002$ & $7.904180760948$ \\
\midrule
\multirow{3}{*}{$1000$} & $0.05$ & \multirow{3}{*}{$7.903106732975$} & \multirow{3}{*}{$0.000273085845$} & $7.902571494554$ & $7.903641971395$ \\
 & $0.01$ &  &  & $7.902403310453$ & $7.903810155496$ \\
 & $0.001$ &  &  & $7.902208136702$ & $7.904005329247$ \\
\midrule
\multirow{3}{*}{$1500$} & $0.05$ & \multirow{3}{*}{$7.903106732975$} & \multirow{3}{*}{$0.000222973659$} & $7.902669712634$ & $7.903543753315$ \\
 & $0.01$ &  &  & $7.902532390891$ & $7.903681075059$ \\
 & $0.001$ &  &  & $7.902373032191$ & $7.903840433759$ \\
\midrule
\multirow{3}{*}{$2000$} & $0.05$ & \multirow{3}{*}{$7.903106732975$} & \multirow{3}{*}{$0.000193100853$} & $7.902728262258$ & $7.903485203692$ \\
 & $0.01$ &  &  & $7.902609338140$ & $7.903604127810$ \\
 & $0.001$ &  &  & $7.902471329457$ & $7.903742136493$ \\
\bottomrule
\end{tabular*}
\end{table}

Algorithm~\ref{alg:localentropy} applies the reference moments of~\eqref{eq:localref} and the acceptance interval of Corollary~\ref{cor:localclt} to a single ciphertext.

\begin{algorithm}
\caption{VSTL $(k, N_B)$-local entropy test}\label{alg:localentropy}
\begin{algorithmic}[1]
\Require ciphertext $S$ consisting of $T$ triangles, with $N = 50T$ data bytes $S_1, \ldots, S_N$; block count $30 \le k \le \left\lfloor \dfrac{T}{39} \right\rfloor$; significance level $\alpha$
\Ensure $(k, 1950)$-local entropy acceptance verdict under the VSTL reference
\State $\mathcal{R} \gets$ the starting positions of $k$ non-overlapping and uniformly selected $1950$-byte blocks
\State $\bar{H} \gets 0$
\For{$i \in \mathcal{R}$}
\State $H_i \gets 0$
\State $\#(v) \gets 0$ for all $v \in \mathcal{V}$
\For{$j \gets 0$ \textbf{to} $1949$}
\State $\#(S_{i+j}) \gets \#(S_{i+j}) + 1$
\EndFor
\For{$v \in \mathcal{V}$ such that $\#(v) > 0$}
\State $H_i \gets H_i - \dfrac{\#(v)}{1950} \log_2 \dfrac{\#(v)}{1950}$ \Comment{\eqref{eq:plugin}}
\EndFor
\State $\bar{H} \gets \bar{H} + \dfrac{H_i}{k}$ \Comment{\eqref{eq:local}}
\EndFor
\State $\mu \gets 7.903106732975$ \Comment{\eqref{eq:localref}}
\State $\sigma \gets \dfrac{0.008635732664}{\sqrt{k}}$ \Comment{\eqref{eq:localref}}
\State $\bar{H}^{-} \gets \mu - \Phi^{-1}(1 - \alpha/2) \sigma$; $\bar{H}^{+} \gets \mu + \Phi^{-1}(1 - \alpha/2) \sigma$ \Comment{\eqref{eq:localinterval}}
\State \Return $\bar{H}^{-} \le \bar{H} \le \bar{H}^{+}$
\end{algorithmic}
\end{algorithm}

\subsection{Practical Significance}\label{subsec:localpractical}

The uniform reference of~\cite[Tab.~3]{local_entropy_2013} uses blocks of $1936$ bytes, whereas the VSTL local entropy test requires blocks of $1950$ bytes.
We therefore compare the VSTL reference with the uniform references at both block sizes, as reported in Table~\ref{tab:localreferences}.

\begin{table}
\caption{Uniform local entropy reference values at $N_B = 1936$ and $N_B = 1950$, compared with the VSTL reference, with crossing thresholds for the mean differences.}
\label{tab:localreferences}
\centering
\footnotesize
\renewcommand{\arraystretch}{1.25}
\setlength{\tabcolsep}{4pt}
\begin{tabular}{lcc}
\toprule
 & Uniform ($N_B = 1936$) & Uniform ($N_B = 1950$) \\
\midrule
$\mu_{\bar{H}}$ & $7.902469316647$ & $7.903192345732$ \\
$\sqrt{k}\,\sigma_{\bar{H}}$ & $0.008694225839$ & $0.008628570032$ \\
$\mu_{\bar{H}}^{\mathrm{VSTL}}-\mu_{\bar{H}}$ & $0.000637416328$ & $-0.000085612757$ \\
$\sigma_{\bar{H}}/\sigma_{\bar{H}}^{\mathrm{VSTL}}$ & $1.006773388890$ & $0.999170582071$ \\
\midrule
Crossing at one standard deviation & $k \ge 187$ & $k \ge 10158$ \\
Crossing at $95\%$ acceptance half-width & $k \ge 715$ & $k \ge 39021$ \\
\bottomrule
\end{tabular}
\end{table}

We express each mean difference in units of the corresponding uniform standard deviation to quantify the separation from the VSTL reference and identify when the uniform reference becomes practically unsuitable:
\begin{equation}\label{eq:driftgrowth}
\left\lvert \mu_{\bar{H}}^{\mathrm{VSTL}} - \mu_{\bar{H}}^{\mathrm{uniform}}(N_B) \right\rvert = z \sigma_{\bar{H}}^{\mathrm{uniform}}(N_B).
\end{equation}
Solving~\eqref{eq:driftgrowth} for $k$ with $z = 1$ gives the crossing at one uniform standard deviation, while with $z = \Phi^{-1}(0.975)$ it gives the crossing at the two-sided $95\%$ acceptance half-width.
The mean differences and corresponding crossings are also listed in Table~\ref{tab:localreferences}.

Against the uniform reference at $N_B = 1936$, both crossings occur at block counts attainable within practically relevant file sizes; at $N_B = 1950$, the crossings occur at substantially larger block counts.

\begin{figure}
\centering
\begin{subfigure}{\linewidth}
\centering
\makebox[\linewidth][l]{\small{\textbf{(a)}}}\\[0pt]
\includegraphics[width=\linewidth]{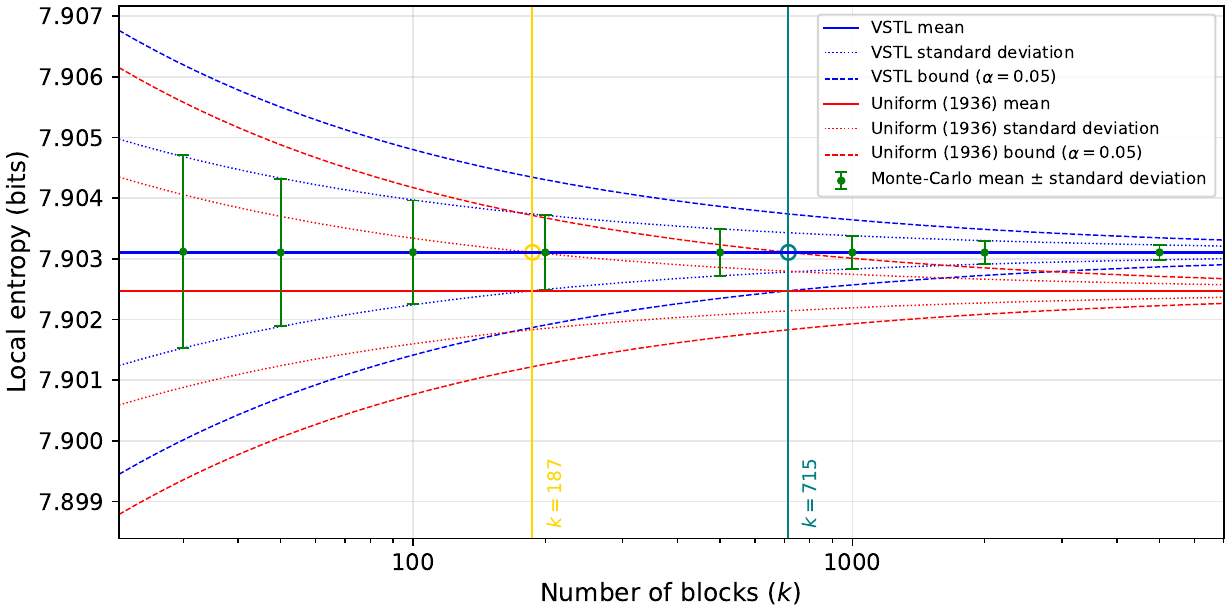}
\phantomsubcaption\label{fig:localmc1936-raw}
\end{subfigure}
\begin{subfigure}{\linewidth}
\centering
\makebox[\linewidth][l]{\small{\textbf{(b)}}}\\[0pt]
\includegraphics[width=\linewidth]{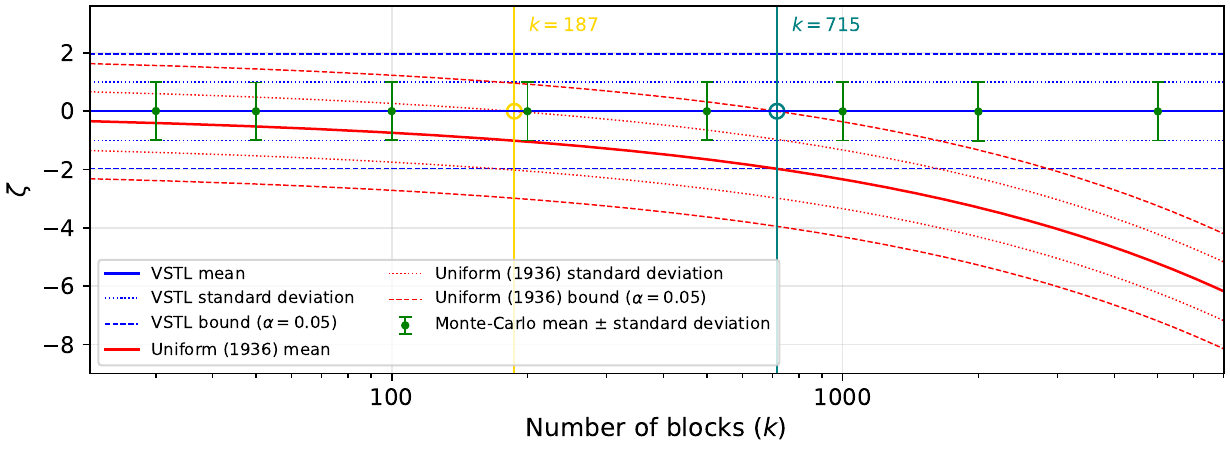}
\phantomsubcaption\label{fig:localmc1936-std}
\end{subfigure}
\caption{Local entropy under the VSTL model, compared with the conventional uniform reference at $N_B = 1936$: (a) exact reference values and Monte-Carlo estimates; (b) the standardized counterpart of (a), expressed in $\zeta$-units.}
\label{fig:localmc1936}
\end{figure}

\begin{figure}
\centering
\begin{subfigure}{\linewidth}
\centering
\makebox[\linewidth][l]{\small{\textbf{(a)}}}\\[0pt]
\includegraphics[width=\linewidth]{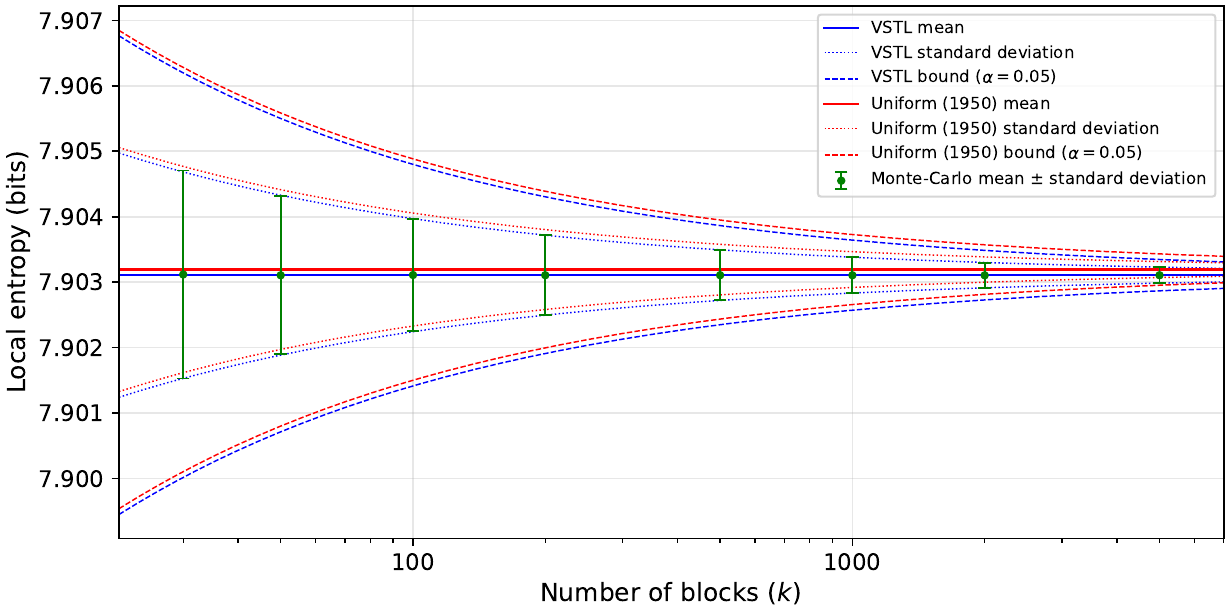}
\phantomsubcaption\label{fig:localmc1950-raw}
\end{subfigure}
\begin{subfigure}{\linewidth}
\centering
\makebox[\linewidth][l]{\small{\textbf{(b)}}}\\[0pt]
\includegraphics[width=\linewidth]{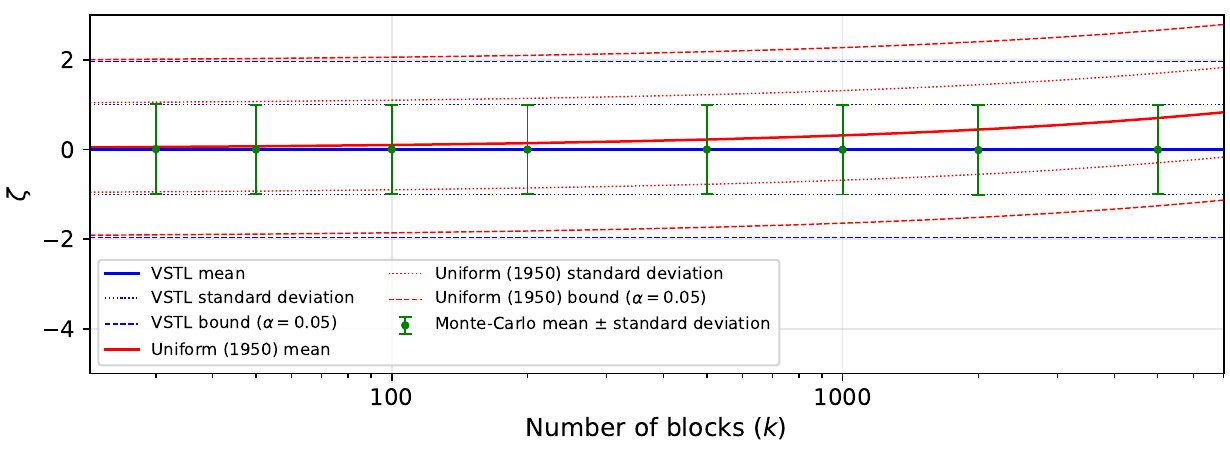}
\phantomsubcaption\label{fig:localmc1950-std}
\end{subfigure}
\caption{Local entropy under the VSTL model, compared with the uniform reference at $N_B = 1950$: (a) exact reference values and Monte-Carlo estimates; (b) the standardized counterpart of (a), expressed in $\zeta$-units.}
\label{fig:localmc1950}
\end{figure}

\subsection{Monte-Carlo Validation}\label{subsec:localmc}

Following the simulation procedure of~\cite{adjusted_npcr_uaci_2026} and Section~\ref{subsec:globalmc}, we check whether samples from the VSTL distribution recover the local entropy reference moments.

For a fixed $k$, one trial draws a VSTL file according to Definition~\ref{def:vstl}, consisting of $k$ blocks of $T_B^{\ast}$ triangle records each.
We compute their average entropy using~\eqref{eq:local}.
By Lemma~\ref{lem:aligned}, this procedure is equivalent in distribution to selecting $k$ disjoint blocks from a longer VSTL file.
We use $n = 10000$ independent trials at each of the selected block counts spanning $30 \le k \le 5000$.

The trial averages give the sample mean $\hat{\mu}_{\bar{H}}$ and unbiased sample variance $\hat{\sigma}_{\bar{H}}^{2}$ through~\eqref{eq:mcestimators}.
Standardization follows~\eqref{eq:zeta}, using the exact local reference values $\mu_{\bar{H}}^{\mathrm{VSTL}}$ and $\sigma_{\bar{H}}^{\mathrm{VSTL}}$ of~\eqref{eq:localref}.

The simulated mean and standard deviation agree with the exact VSTL predictions within sampling error across all tested block counts, as illustrated in Figures~\ref{fig:localmc1936} and~\ref{fig:localmc1950}.
In the comparison with the uniform reference at $N_B = 1936$, the fixed difference between the means becomes increasingly pronounced as the acceptance intervals narrow for larger $k$.
The vertical lines in Figure~\ref{fig:localmc1936} mark the crossings derived in Section~\ref{subsec:localpractical} and listed in Table~\ref{tab:localreferences}; beyond the second crossing, the VSTL reference mean lies above the upper bound of the uniform $95\%$ acceptance interval.
For the uniform reference evaluated at the optimal block size for the VSTL local entropy test $N_B^{\ast}$, the mean difference is smaller and the curves are substantially closer, as shown in Figure~\ref{fig:localmc1950}, with both crossing thresholds lying beyond the tested range.

Figure~\ref{fig:localdiff} shows the absolute differences between the reference values and their Monte-Carlo estimates on a logarithmic scale; both show an overall downward trend towards zero as the block count grows.

\begin{figure}
\centering
\includegraphics[width=\linewidth]{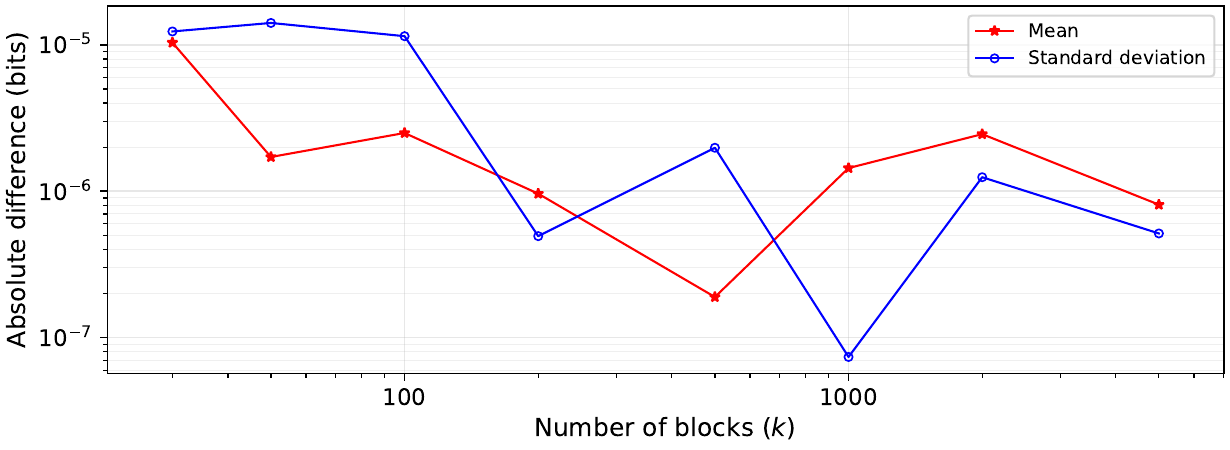}
\caption{Absolute difference between the reference values of the VSTL local entropy test and their corresponding Monte-Carlo estimates: $\lvert \mu_{\bar{H}}^{\mathrm{VSTL}} - \hat{\mu}_{\bar{H}} \rvert$ (red stars) and $\lvert \sigma_{\bar{H}}^{\mathrm{VSTL}} - \hat{\sigma}_{\bar{H}} \rvert$ (blue circles).}
\label{fig:localdiff}
\end{figure}

\section{Conclusions}\label{sec:conclusions}

This paper introduced a two-sided hypothesis test for the global entropy under the VSTL null model and adjusted the established local entropy test to this setting.
The reference values and acceptance bounds account for the IEEE-754 validity constraint and the resulting dependence between the exponent bytes of each float.
Removing this validity constraint recovers the conventional uniform model.

The global entropy converges to approximately $7.99991479609$ bits, with an expected value strictly below this limit for every finite sample size.
The second-order expansions provide efficiently computable reference moments and retain a nonzero $N^{-1}$ term in the variance, unlike the uniform case, where this term vanishes.
The difference between the reference means becomes relevant at practical file sizes --- at $3824$ triangles it first exceeds one uniform standard deviation, while at $7495$ triangles the VSTL reference mean first falls below the lower bound of the uniform $95\%$ acceptance interval.

For local entropy, the capacity energy criterion increased the block size from the conventional $1936$ bytes to $1950$ bytes, corresponding to $39$ triangles.
The exact moments gave a reference mean of approximately $7.903106732975$ bits and a standard deviation of approximately $\dfrac{0.008635732664}{\sqrt{k}}$ for the average over $k$ disjoint blocks.
The difference from the uniform reference at $1936$ bytes is again practically relevant --- at $187$ blocks it first exceeds one uniform standard deviation, while at $715$ blocks the VSTL reference mean first exceeds the upper bound of the uniform $95\%$ acceptance interval; the reference at $1950$ bytes is substantially closer, but still becomes unsuitable for larger numbers of blocks.

Monte-Carlo simulations confirm the global and local reference moments within sampling error in all tested cases.

We therefore recommend evaluating STL ciphers against the VSTL references derived in this paper.
Acceptance indicates consistency with the stated null model in terms of entropy; it does not establish cryptographic security.

\section*{Declaration of competing interest}

The authors declare that they have no known competing financial interests or personal relationships that could have appeared to influence the work reported in this paper.

\section*{Funding}

This work was supported by the Scientific Research Fund of Konstantin Preslavsky University of Shumen [grant number RD-08-91/02.02.2026].

\section*{Declaration of generative AI and AI-assisted technologies in the manuscript preparation process}

Generative AI tools were used to assist in identifying related literature, improve the language and readability of the manuscript, proofread the text, review the mathematical derivations and validate the consistency of the reported results.
The authors retain full responsibility for the correctness of this work.

\bibliographystyle{elsarticle-num}
\bibliography{references}

\end{document}